\documentclass[journal]{IEEEtran}
\IEEEoverridecommandlockouts
\usepackage{cite}
\usepackage{graphicx}
\graphicspath{{.../FIG/}}
\DeclareGraphicsExtensions{.pdf,.jpeg,.png,.webp,.jpg}
\usepackage{stfloats}
\usepackage[hidelinks]{hyperref}
\usepackage{amsmath,mathtools}
\usepackage{amsfonts,amssymb}
\usepackage{bm}
\usepackage{algorithm}
\usepackage{algorithmic}

\usepackage{array}
\renewcommand\arraystretch{1}
\usepackage{cuted} 
\usepackage{booktabs}
\usepackage{multirow}
\usepackage{tabularx}
\usepackage{longtable}
\usepackage{pifont}
\usepackage{mathrsfs}

\usepackage{url}

\usepackage{amsthm}

\newtheorem{remark}{Remark}

\newtheorem{example}{Example}
\newtheorem{theorem}{Theorem}[section]

\newtheorem{definition}[theorem]{Definition}
\newtheorem{corollary}[theorem]{Corollary}
\newtheorem{proposition}[theorem]{Proposition}
\usepackage{color}
\definecolor{BLUE}{RGB}{0,114,189}
\definecolor{RED}{RGB}{217,83,25}
\definecolor{YELLOW}{RGB}{237,177,32}
\definecolor{PURPLE}{RGB}{126,47,142}

\usepackage{lipsum}
\usepackage{amssymb}
\usepackage{booktabs}
\usepackage{enumerate}
\usepackage{setspace}
\usepackage{float}

\usepackage[caption=false,font=normalsize,labelfont=sf,textfont=sf]{subfig}
\usepackage{textcomp}
\usepackage{verbatim}

\usepackage{subfiles}

\usepackage{soul}
\usepackage{textcomp}
\usepackage{xcolor}
\definecolor{myblue}{RGB}{0,112,192}

\begin{document}

\title{Geometric Decentralized Stability Certificate of Power Electronics-Dominated Power Systems Covering Variable Operating Points} % <-this % stops a space

\author{Ruohan~Leng,~Linbin~Huang,~Liangxiao~Luo,~Huanhai~Xin,~Xiongfei~Wang,~and~Florian~Dörfler% <-this % stops a space
\vspace{-2.5mm}
\thanks{This work was supported by the National Natural
Science Foundation of China under Grant U24B6008. }
\thanks{Ruohan Leng, Linbin Huang, Liangxiao Luo, and Huanhai Xin are with the College of Electrical Engineering, Zhejiang University, Hangzhou 310027, China. (e-mail: {\{lengruohan, hlinbin, luolx, xinhh\}}@zju.edu.cn).}\\
\thanks{Xiongfei Wang is with the Department of Electrical Engineering, Tsinghua University, Beijing, China. (e-mail: xiongfei@tsinghua.edu.cn).}
\thanks{Florian D{\"o}rfler is with the Department of Information Technology and Electrical Engineering at ETH Z{\"u}rich, Switzerland. 
(e-mail: dorfler@ethz.ch).}
}% <-this % stops a space
\maketitle

%TODO *** ABSTRACT ***
\begin{abstract}

The integration of power converters is profoundly changing the power system dynamics and poses significant challenges for stability analysis. The dynamic interactions between the power grid and the heterogeneous converters are highly complex and difficult to analyze due to the curse of dimensionality. Moreover, system stability varies with the operating points, which are determined by the voltage magnitude, active power, and reactive power of each converter. This further complicates the analysis as it is difficult to enumerate and examine all the possible operating points. To tackle these challenges, this paper proposes a geometric decentralized stability certificate for power electronics (PE)-dominated power systems, which can simultaneously handle heterogeneous power converters and their variable operating points. The certificate can be checked in a decentralized and modular manner, and it is scalable for large-scale power systems.
Our approach is developed based on the concept of Davis-Wielandt (DW) shell and its projections, which can effectively visualize the characteristics of high-dimensional complex matrices. We investigate how the projections of the DW shell vary with operating points and how this variation can guide the search for worst-case operating conditions. We further propose an efficient algorithm to compute the stability margin and construct the certified operating regions. The effectiveness of the proposed method is validated through case studies on single-converter and 54-converter wind power systems.
%We test the effectiveness of the method by simulating a wind power system.

\end{abstract}

\vspace{0mm}

\vspace{-3mm}

%TODO *** 1 INTRODUCTION ***
\section{Introduction}

Due to the developments of renewable energy generation, high-voltage DC transmission, and energy storage systems, modern power systems are evolving toward power electronics (PE)-dominated power systems where PE converters are widely used for energy conversion~\cite{milano2018foundations,kroposki2017achieving,lin2025enhancing}. The stability analysis of PE-dominated power systems has long been considered an intricate task due to the complex dynamic interactions between the power grid and heterogeneous converters. Moreover, since power systems contain nonlinear dynamics, their small-signal stability highly depends on the operating points~\cite{zhang2021artificial}, i.e., the active and reactive power of the converters. For instance, it has been shown that a converter often has a lower stability margin at its maximum active power output~\cite{liu2024generalized}.

The operating points of a PE-dominated power system are difficult to enumerate due to the high-dimensional combinatorial nature of all converters' active and reactive power outputs. Hence, it is computationally challenging to examine system stability over all possible operating points~\cite{zou2020two}. It remains an open question how to ensure that a PE-dominated power system remains stable under all possible operating points.

Existing studies have investigated the impact of operating points on small-signal stability from several perspectives. For a single-converter system, instability may occur at low power outputs, indicating that stability can vary non-monotonically with the operating points~\cite{hossen2023hidden}. Full-operating-space analysis has also been developed by directly embedding the variables of operating points into the small-signal model and constructing a safe operating region~\cite{xie2025analysis}. For multi-converter systems with homogeneous devices, variable-operating-point impedance models~\cite{liu2023stability} and wind-speed-dependent admittance models~\cite{du2023stability} have been proposed to extend stability analysis to a wide range of operating conditions. Additionally, stability regions have also been constructed in the injection space to assess stability under varying power outputs~\cite{zhan2024injection}. However, these approaches typically focus on homogeneous converter dynamics or avoid the high-dimensional operating-point space through restrictive assumptions, and are therefore not scalable to heterogeneous PE-dominated power systems.
%require exploring the high-dimensional operating-point space in a combinatorial manner, and thus become computationally prohibitive for heterogeneous PE-dominated power systems.

Recently, the concept of decentralized stability certificates (DSCs) has gained considerable attention because it enables modular and scalable analysis of power system dynamics~\cite{vorobev2019decentralized,siahaan2024decentralized,chen2025unified}. The traditional passivity condition~\cite{moylan1977stability} can be viewed as an early example of DSCs, as system stability can be checked in a decentralized way: the system is stable if all the interconnected components are passive. However, the passivity condition usually does not hold over the whole frequency range~\cite{chen2025extended}, especially when synchronization dynamics are considered. In~\cite{huang2020h}, a DSC was developed based on the small-gain theorem, which can be further combined with the small-phase theorem to reduce the conservativeness, as investigated in~\cite{huang2024gain,cifelli2027decentralized}. The small-phase theorem extends the passivity condition by explicitly defining the matrix phases based on the numerical range, and the passivity condition can be considered as constraining the phases within $\pm 90^\circ$~\cite{chen2024phase}. There are also DSCs derived from the scaled relative graph (SRG), Davis-Wielandt (DW) shell and its projections~\cite{baron2025decentralized, feng2026unified, huang2025geometric}. 

These DSCs enable scalable system stability assessment by checking whether the characteristics of each device match those of the power network. This avoids constructing high-dimensional state-space matrices, supports interoperability among heterogeneous devices, and provides insights into device-grid interactions while reducing the need for full sharing of proprietary device models. As demonstrated in~\cite{huang2024gain}, the characteristics of the devices and the power network can be described by their admittance matrices, and the effect of operating points is reflected only in the device admittances. Since DSCs focus on the individual admittance matrix of each device, they break the combinatorial problem of all the devices' active and reactive power into the individual device's active and reactive power for stability analysis. For instance, Ref.~\cite{baron2027impact} accounts for the operating points using parameterized SRGs with an affine structure. However, such an affine structure has limitations in describing how different active and reactive power affect the devices' admittance matrix. 

This paper aims to develop an effective method to analyze and guarantee the small-signal stability of PE-dominated power systems under variable operating points. As a first step, we develop a DSC based on a unified DW-shell geometric viewpoint, covering the small-gain condition, the small-phase condition, the $x$-$z$ graph (the projection of DW shell onto the $x$-$z$ plane), and SRG. Our certificate extends the one in our previous work~\cite{huang2025geometric} by incorporating SRG, which can leverage the advantages of different geometric concepts. Unlike~\cite{baron2027impact} which assumes an affine structure, we consider a general dependence of the converter's admittance matrix on the operating points. Then, we define a geometric stability margin and the corresponding decentralized stability operating region based on the $x$-$z$ graph. Moreover, we develop an efficient algorithm to compute the stability margin and construct the certified operating region. We further investigate how the converter $x$-$z$ graph varies with operating points and use this geometric dependence to guide the search for worst-case operating conditions. Our method enables systematic stability assessment of PE-dominated power systems over admissible operating-point sets and provides guidance for converter control design.

The rest of this paper is organized as follows. Section~\ref{2} establishes the operating-point-dependent model of PE-dominated power systems. Section~\ref{3} develops the geometric decentralized stability certificate based on geometric graphs of complex matrices. Section~\ref{4} defines the decentralized stability operating region and presents its efficient construction and worst-case analysis. Section~\ref{5} validates the proposed method through simulations. Section~\ref{6} concludes the paper.

\section{Operating-Point-Dependent Modeling \\ of PE-Dominated Power Systems}\label{2}
In this section, we first specify the operating points of PE-dominated power systems, and then derive the operating-point-dependent admittance models of converters under PQ and PV control modes. The network dynamics are also derived and interconnected with the converter dynamics to obtain the closed-loop representation of the system.

\subsection{Operating Points of PE-Dominated Power Systems}

Consider a typical PE-dominated power system in which $N$ ($N \in \mathbb{Z}_{>0}$) converters are interconnected through the power network, as illustrated in Fig.~\ref{common_topology}. The converters may employ different control structures and parameters, and operate at different steady-state operating points. The operating point of the $i$-th converter is defined as
\begin{equation}
\mathcal{O}_i := [P_i,Q_i,V_i]^\top \in \mathbb{R}^3,
\end{equation}
where $P_i$, $Q_i$, and $V_i$ denote the active power, reactive power, and terminal voltage magnitude, respectively. The admissible set of operating points of the \(i\)-th converter is defined as
\begin{equation}
\label{eq:Omega_i}
\bar{\mathcal S}_i
\!:=\!
\left\{
\mathcal O_i\!:\!
P_i^2+Q_i^2\le \kappa_i^2,\ \!
\angle(P_i\!+\!\mathrm j Q_i)\!\in\!\Psi_i,\ \!
V_i\!\in\!\mathcal V_i
\right\},
\end{equation}
where \(\kappa_i\), \(\Psi_i\), and \(\mathcal V_i\) specify the admissible power-injection magnitude limit, angle range, and voltage set, respectively.

The operating point of the entire system is defined as
\begin{equation}
\label{eq:system_operating_point}
\mathcal O
:=
[\mathcal O_1^\top,\ldots,\mathcal O_N^\top]^\top
\in \mathbb R^{3N}.
\end{equation}
Accordingly, the admissible set of system operating points is
\begin{equation}
\label{eq:system_admissible_set}
\bar{\mathcal S}
:=
\left\{
\mathcal O:
\mathcal O_i\in\bar{\mathcal S}_i,\ i=1,\ldots,N
\right\}
=
\bar{\mathcal S}_1\times\cdots\times\bar{\mathcal S}_N .
\end{equation}

The small-signal stability of a multi-converter system depends on its operating point \(\mathcal{O}\), and the stability margin also changes with the operating point. For instance, the system may become unstable with variations of the operating point. However, it is difficult to examine each possible operating point in \(\bar{\mathcal{S}}\) due to the exploding dimension of this \(3N\)-dimensional operating-point space as the number of converters increases. This motivates the following decentralized stability analysis, where we develop stability conditions for each converter to satisfy. Our approach avoids directly dealing with the high-dimensional \(\mathcal{O}\), by instead focusing on each \(\mathcal{O}_i\).

\begin{figure}[!t]
	\centering
	\includegraphics[width=0.85\linewidth]{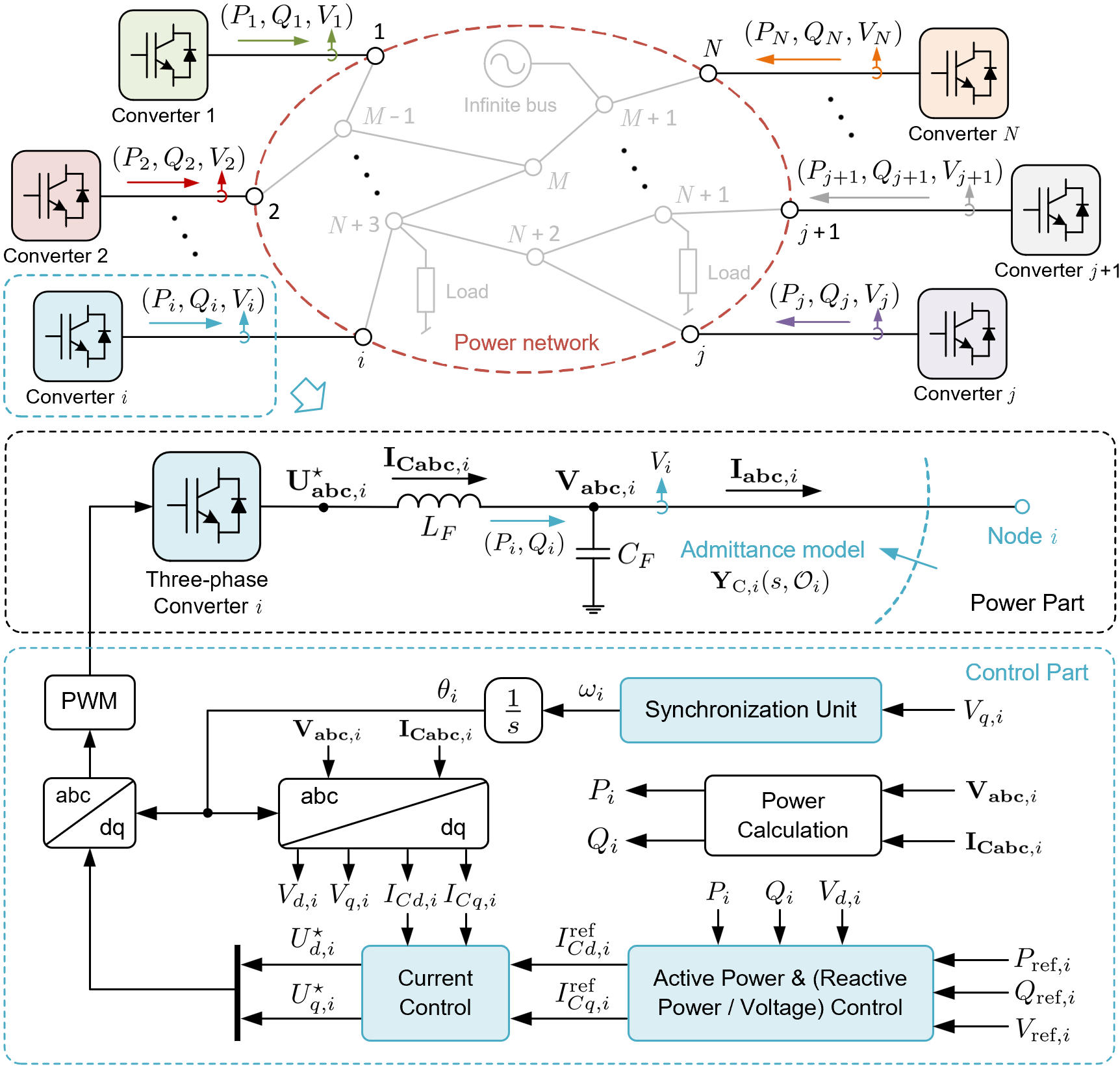}
	\vspace{-2mm}
	%\DeclareGraphicsExtensions.
	\caption{Illustrations of a PE-dominated power system and the control structure of the \(i\)-th converter. The converter can operate either in PQ or PV control mode, and the synchronization unit can be a PLL~\protect\cite{huang2024gain}.}
	\vspace{-4mm}
	\label{common_topology}
\end{figure}

\subsection{Operating-Point-Dependent Converter Dynamics}
By linearizing around an operating point, the small-signal dynamics of the $i$-th converter can be represented by the following admittance model in the frequency domain~\cite{yang2020placing,xin2025how}
\begin{equation}\label{eq:admittance_model}
    -\begin{bmatrix} I_{d,i} \\ I_{q,i} \end{bmatrix} = S_i e^{J\theta_i} {\bf Y}_{{\rm C},i}(s,\mathcal{O}_i) e^{-J\theta_i}\begin{bmatrix} V_{d,i} \\ V_{q,i} \end{bmatrix} \,,
\end{equation}
where $[ I_{d,i} , I_{q,i} ]^\top$ and $[ V_{d,i} , V_{q,i} ]^\top$ denote the output current and voltage \textit{perturbations} of the $i$-th converter in a global $dq$ coordinate, and $S_i$ represents the rated capacity of the $i$-th converter. The {steady-state angle difference} between the global $dq$ coordinate and converter’s local $dq$ coordinate is denoted by $\theta_i$, and $\renewcommand{\arraystretch}{0.7}
e^{J\theta_i} =
\begin{bmatrix}
{\scriptstyle \cos \theta_i} \ & {\scriptstyle -\sin \theta_i} \\
{\scriptstyle \sin \theta_i} \ & {\scriptstyle \cos \theta_i}
\end{bmatrix}$.
The $2\times2$ admittance matrix ${\bf Y}_{{\rm C},i}(s,\mathcal{O}_i)$ 
describes the converter dynamics in the local $dq$ coordinate; see~\cite{yang2020placing} and~\cite{xin2025how} for the detailed derivations.

\begin{figure*}[!t]
\vspace{-8mm}
\begin{equation}
\begin{split}
&{\bf Y}_{\mathrm{C},i}^{\rm PQ}(s,\mathcal{O}_i)
=
\begin{bmatrix}\,
\tfrac{G_{I,i}
(s) \mathrm{PI}_{\mathrm{AC},i}(s) P_i + V_i Y_{\mathrm{VF},i}(s)}{V_i D_{p,i}(s)} &
-\tfrac{Q_i\left[\,s G_{I,i}(s) \mathrm{PI}_{\mathrm{AC},i}(s)
+\mathrm{PI}_{\mathrm{PLL},i}(s) D_{p,i}(s)\,\right]}%
{V_i D_{p,i}(s) S_i(s)} \\[8pt]
-\tfrac{Q_i G_{I,i}(s) \mathrm{PI}_{\mathrm{RC},i}(s)}{V_i D_{q,i}(s)} &
\quad\tfrac{s V_i Y_{\mathrm{VF},i}(s)
- P_i\left[\, s G_{I,i}(s) \mathrm{PI}_{\mathrm{RC},i}(s)
+\mathrm{PI}_{\mathrm{PLL},i}(s) D_{q,i}(s)\,\right]}%
{V_i S_i(s) D_{q,i}(s)}
\,\end{bmatrix} + {\bf Y}_{\mathrm{F},i}(s),
\\[5pt]
&{\bf Y}_{\mathrm{C},i}^{\rm PV}(s,\mathcal{O}_i)
=
\begin{bmatrix}
\,\tfrac{G_{I,i}(s) \mathrm{PI}_{\mathrm{AC},i}(s) P_i + V_i Y_{\mathrm{VF},i}(s)}{V_i D_{p,i}(s)} & \quad\quad
-\tfrac{Q_i\left[\,s G_{I,i}(s) \mathrm{PI}_{\mathrm{AC},i}(s)
+\mathrm{PI}_{\mathrm{PLL},i}(s) D_{p,i}(s)\,\right]}%
{V_i D_{p,i}(s) S_i(s)} \,\,\\[8pt]
-\,G_{I,i}(s)\,\mathrm{PI}_{\mathrm{VC},i}(s) &
\tfrac{s V_i Y_{\mathrm{VF},i}(s) - \mathrm{PI}_{\mathrm{PLL},i}(s) P_i}{V_i S_i(s)}
\end{bmatrix} + {\bf Y}_{\mathrm{F},i}(s).
\end{split}
\label{eq:Y0_PQ_PV}
\end{equation}
\vspace{-6mm}
\end{figure*}

The expression of ${\bf Y}_{{\rm C},i}(s,\mathcal{O}_i)$ is determined by the converter's AC-side dynamics and control architecture, and depends on the steady states of the $dq$-axis voltage and current, which are $V_{d0,i}$, $V_{q0,i}$, $I_{d0,i}$, and $I_{q0,i}$. We then explicitly parameterize ${\bf Y}_{{\rm C},i}(s,\mathcal{O}_i)$ as a function of $\mathcal{O}_i$ by considering that $P_i = V_{d0,i} I_{d0,i}$, $Q_i = -V_{d0,i} I_{q0,i}$, and $V_i = V_{d0,i}$, since $V_{q0,i}=0$ under steady-state conditions. The resulting expressions of ${\bf Y}_{{\rm C},i}(s,\mathcal{O}_i)$ are given in~\eqref{eq:Y0_PQ_PV} as ${\bf Y}_{{\rm C},i}^{\rm PQ}(s,\mathcal{O}_i)$ and ${\bf Y}_{{\rm C},i}^{\rm PV}(s,\mathcal{O}_i)$, corresponding to the PQ and PV control modes, respectively. Note that, for fixed \(V_i\), both expressions are affine in \((P_i,Q_i)\), which will be used later in Section~\ref{4}. Here, $G_{I,i}(s)$ represents the current-loop tracking dynamics and $Y_{\mathrm{VF},i}(s)$ the voltage-feedforward dynamics; $\mathrm{PI}_{\mathrm{AC},i}(s)$, $\mathrm{PI}_{\mathrm{RC},i}(s)$, $\mathrm{PI}_{\mathrm{VC},i}(s)$, and $\mathrm{PI}_{\mathrm{PLL},i}(s)$ denote the PI controllers for active-power, reactive-power, voltage, and PLL control, respectively. For compactness, we define the common denominators $D_{p,i}(s)=1+G_{I,i}(s)\mathrm{PI}_{\mathrm{AC},i}(s)V_i$, $D_{q,i}(s)=1+G_{I,i}(s)\mathrm{PI}_{\mathrm{RC},i}(s)V_i$, and $S_i(s)=s+\mathrm{PI}_{\mathrm{PLL},i}(s)V_i$; $\renewcommand{\arraystretch}{0.7}{\bf Y}_{\mathrm{F},i}(s)=\tfrac{C_{F,i}}{\omega_0}\!
\begin{bmatrix}
{\scriptstyle s} \ & {\scriptstyle -\omega_0} \\
{\scriptstyle \omega_0} \ & {\scriptstyle s}
\end{bmatrix}$ captures the dynamics of the filter capacitor $C_{F,i}$, where $\omega_0$ is the nominal frequency.

  \subsection{Closed-Loop Modeling of PE-Dominated Power System}
Based on the $i$-th converter admittance model in~\eqref{eq:admittance_model} and its operating-point-dependent reformulation ${\bf Y}_{{\rm C},i}(s,\mathcal{O}_i)$, we extend it to include the dynamics of all $N$ converters. For ease of notation, let
${\bf I}=[I_{d,1},I_{q,1},\dots,I_{d,N},I_{q,N}]^\top$ and
${\bf V}=[V_{d,1},V_{q,1},\dots,V_{d,N},V_{q,N}]^\top$, respectively. Then, we have
\begin{equation}\label{eq:admittance_devices}
    -{\bf I}
    =
    [\,
    ({\bf S}\otimes I_2)\,
    e^{J{\boldsymbol\theta}}\,
    {\bf Y}_{\mathrm{C}}^{N}(s, \mathcal{O})\,
    e^{-J{\boldsymbol\theta}}
    \,]
    {\bf V} ,
\end{equation}
where ${\bf S} = {\rm diag}\{S_1,\dots,S_N\} \in \mathbb{R}^{N\times N}$ 
is the rated capacity matrix, and 
$e^{J{\boldsymbol\theta}} = {\rm diag}\{e^{J\theta_1},\dots,e^{J\theta_N}\} 
\in \mathbb{R}^{2N\times 2N}$ represents the steady-state coordinate transformation. 
The block-diagonal transfer matrix 
${\bf Y}_{\mathrm{C}}^{N}(s, \mathcal{O}) \in \mathbb{R}(s)^{2N\times 2N}$ is given by
\begin{equation*}
\scalebox{1.0}{$
\begin{aligned}
{\bf Y}_{\mathrm{C}}^{N}(s, \mathcal{O})
 =
&\begin{bmatrix}
{\bf Y}_{{\rm C},1}(s,\mathcal{O}_1) & & \\
& \ddots & \\
& & {\bf Y}_{{\rm C},N}(s,\mathcal{O}_N)
\end{bmatrix}.
\end{aligned}
$}
\end{equation*}
%extended from ${\bf Y}_{{\rm C},i}(s,\mathcal{O}_i)$,

We next develop the transfer function matrix of the power network. Consider the network shown in Fig.~\ref{common_topology}, which consists of $N$ converter nodes, denoted by Nodes~$1 \sim N$, $M-N$ interior nodes ($M \in \mathbb{Z}_{>0},\, M \ge N$), denoted by Nodes~$N+1 \sim M$, and one common grounded node, denoted by Node~$M+1$. The dynamic equation of the line connecting Node~$i$ and Node~$j$ ($i,j \in \{1,\dots,M+1\}$) is
\begin{equation}\label{eq:edge_ij}
    \begin{bmatrix} I_{d,ij} \\ I_{q,ij} \end{bmatrix}
    = {\bf Y}_{ij}(s)
    \left(
        \begin{bmatrix} V_{d,i} \\ V_{q,i} \end{bmatrix}
        -
        \begin{bmatrix} V_{d,j} \\ V_{q,j} \end{bmatrix}
    \right),
\end{equation}
where $[I_{d,ij}, I_{q,ij}]^\top$ and $[V_{d,i}, V_{q,i}]^\top$ are the current and voltage vectors, respectively, defined in the global $dq$ frame, and ${\bf Y}_{ij}(s)$ is a $2\times2$ transfer function matrix characterizing, for instance, the line dynamics, which can be expressed as
\begin{equation}\label{eq:Y_ij}
    {\bf Y}_{ij}(s)
    = B_{ij}
      \begin{bmatrix}
        s/\omega_0 + \epsilon & -1 \\[1mm]
        1 & s/\omega_0 + \epsilon
      \end{bmatrix}^{-1}
    =: B_{ij} {\bf F}_{\epsilon}(s),
\end{equation}
where $B_{ij} = 1/X_{ij}$ is the line susceptance with $X_{ij}$ being the reactance, and $\epsilon$ denotes the network R/X ratio.

By assembling the line dynamics in~\eqref{eq:Y_ij}, the network admittance can be expressed via its grounded Laplacian matrix ${\bf B} \!\in\! \mathbb{R}^{M \times M}$, which encodes the network topology and line susceptances, where ${\bf B}_{ij} = - B_{ij}$ for $i \ne j$, and ${\bf B}_{ii} = \sum_{j=1}^{M+1} B_{ij}$. By eliminating the $M\!-\!N$ interior nodes via Kron reduction~\cite{dorfler2013kron},
we obtain the reduced Laplacian matrix ${\bf B}_{\rm r} \in \mathbb{R}^{N \times N}$.
The resulting network dynamics are described by the following
$2N \times 2N$ transfer function matrix:
\begin{equation}\label{eq:Y_grid}
    {\bf I} = [\;{\bf B}_{\rm r} \otimes{\bf F}_\epsilon(s)\; ] {\bf V} =:{\bf Y}_{\rm grid}(s) {\bf V},
\end{equation}
where ${\bf B}_{\rm r}$ is given by ${\bf B}_{\rm r} = {\bf B}_1 - {\bf B}_2{\bf B}_4^{-1} {\bf B}_3$, and 
    $$
        {\bf B}=: \begin{bmatrix} {\bf B}_1 \in \mathbb{R}^{N \times N} \quad & {\bf B}_2 \in \mathbb{R}^{N \times (M-N)} \\ {\bf B}_3 \in \mathbb{R}^{(M-N) \times N} \quad & {\bf B}_4 \in \mathbb{R}^{(M-N) \times (M-N)}\end{bmatrix}.
    $$

Following~\cite{huang2020h}, the small-signal dynamics of the multi-converter system can be represented by the closed-loop interaction between the converters’ dynamics in~\eqref{eq:admittance_devices} and the power network dynamics in~\eqref{eq:Y_grid} as 
\begin{equation}
    [\,(\mathbf S\!\otimes\! I_2)\,e^{J{\boldsymbol\theta}}\,
    \mathbf Y_{\mathrm C}^{N}(s,\mathcal{O})\,
    e^{-J{\boldsymbol\theta}}\,]\;\#\; \mathbf Y_{\rm grid}^{-1}(s),
\end{equation}
where \(\#\) denotes the feedback interconnection.

For the subsequent analysis, we re-scale the converters’ dynamics and the network dynamics following~\cite{huang2024gain} and obtain the equivalent closed-loop dynamics as
\begin{equation}\label{eq:rescaled_sys}
    [\,e^{J{\boldsymbol\theta}}\,\widetilde{\mathbf Y}_{\mathrm C}^{N}(s,\mathcal{O})\, e^{-J{\boldsymbol\theta}}\,]\;\#\; \widetilde{\mathbf Y}_{\rm grid}^{-1}(s),
\end{equation}
where the \(i\)-th block of \(\widetilde{\mathbf Y}_{\mathrm C}^{N}(s,\mathcal{O})\) is defined as
\begin{equation}\label{eq:conv_dy}
    \widetilde{\mathbf Y}_{{\mathrm C},i}(s,\mathcal{O}_i)
    := \mathbf Y_{{\mathrm C},i}(s,\mathcal{O}_i)\,{\bf F}_\epsilon^{-1}(s),
\end{equation}
for \(i \in \{1,\ldots,N\}\), and
\begin{equation}\label{eq:grid_dy}
    \widetilde{\mathbf Y}_{\rm grid}(s)
    = \mathbf S^{-1/2}\,\mathbf B_{\rm r}\,\mathbf S^{-1/2}\,\otimes I_2.
\end{equation}

As illustrated in Fig.~\ref{Fig_closed_loopYY}, 
the closed-loop system in~\eqref{eq:rescaled_sys} features a block-diagonal structure in $\widetilde{\mathbf Y}_{\mathrm C}^{N}(s,\mathcal{O})$, while $\widetilde{\mathbf Y}_{\rm grid}(s)$ is a positive definite constant matrix (independent of~$s$, and hereinafter denoted simply as $\widetilde{\mathbf Y}_{\rm grid}$). These properties facilitate the subsequent decentralized stability analysis.

    \begin{figure}[!t]
    \vspace{0mm}
	\centering
	\includegraphics[width=2.6in]{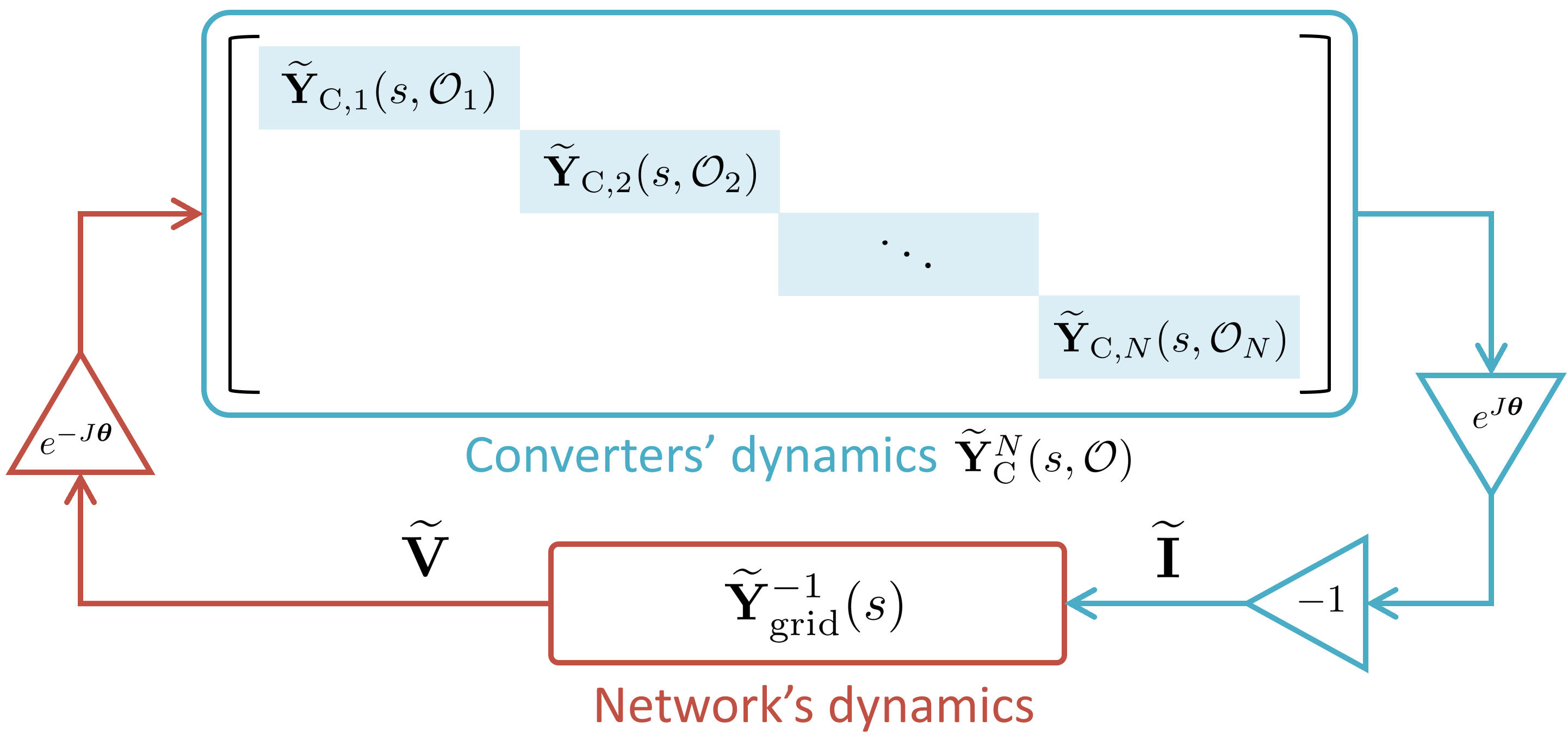}
	\vspace{-2mm}
	%\DeclareGraphicsExtensions.
	\caption{Closed-loop interaction between converters and power network.}
	\vspace{-4mm}
	\label{Fig_closed_loopYY}
    \end{figure}

\section{Geometric Decentralized Stability Certificate for PE-Dominated Power Systems} \label{3}
In this section, we first recall several geometric graphs of matrices, including the numerical range, SRG, DW shell and its $x$-$z$ projection. Building on these graphs, we then develop a geometric decentralized stability certificate for PE-dominated power systems under variable operating points.

\subsection{Geometric Features of Complex Matrices}
The gain and phase of transfer functions are classical frequency-domain quantities for certifying feedback stability~\cite{huang2024gain}, as exemplified by the small-gain~\cite{huang2020h} and passivity conditions~\cite{chen2025extended}. For matrix-valued dynamics, these quantities and their generalizations can be described by geometric graphs of matrices, which are briefly introduced below.

The \textbf{numerical range} of a complex matrix $A \in \mathbb{C}^{n \times n}$ is 
\begin{equation}\label{eq:WA}
W(A)=\{v^*Av: v\in \mathbb{C}^n, \|v\|=1\},
\end{equation}
which is a convex subset of $\mathbb{C}$ containing the eigenvalues of $A$. Here, $v^*$ denotes the conjugate transpose of $v$, and $\|v\|$ denotes the two-norm. If $0\notin W(A)$, then $A$ is \textit{sectorial}. For a sectorial $A$, there exists a nonsingular matrix $T$ and a diagonal unitary matrix $D$ such that $A=T^*DT$, referred to as the sectorial decomposition. Here $D$ is unique up to permutation, with all diagonal entries lying on an arc of the unit circle with length smaller than $\pi$. Then, the \textit{phases} of $A \in\mathbb{C}^{n\times n}$, denoted by
\begin{equation}
    \overline\phi(A)=\phi_1(A) \geq \dots \geq \phi_n(A)=\underline\phi(A), 
    \tag{\bf Phases}
\end{equation}
are defined as the phases of the eigenvalues of $D$ so that $\overline\phi(A)-\underline\phi(A)<\pi$~\cite{chen2024phase,wang2020phases}. The smallest phase $\underline\phi(A)$ and the largest phase $\overline{\phi}(A)$ of a sectorial matrix $A$ are the two supporting rays of $W(A)$, and the other phases lie in between.

As a counterpart, the gains (magnitudes) of a complex matrix $A\in\mathbb{C}^{n\times n}$ can be defined by its singular values
\begin{equation}
    \overline{\sigma}(A) =\sigma_1(A) \geq \dots \geq \sigma_n(A)=\underline{\sigma}(A).
    \tag{\bf Gains}
\end{equation}

The above concepts of ``phases'' and ``gains'' describe the properties of a complex matrix from different perspectives, and lead to different stability criteria of feedback interconnected dynamical systems, namely, the small gain theorem and the small phase theorem~\cite{skogestad2005multivariable, chen2024phase}. In recent years, the concept of the DW shell, tracing back to the works by Davis and Wielandt in~\cite{wielandt1955eigenvalues} and~\cite{davis1968shell}, has gained renewed interest thanks to its ability to simultaneously describe the gain and phase features of a complex matrix~\cite{zhang2025phantom,lestas_DW,Li_DW}.
To be specific, the \textbf{DW shell} of a complex matrix $A \in \mathbb{C}^{n \times n}$ is
\begin{equation}\label{eq:DWA}
DW(A)=\{(v^*Av,\|Av\|^2): v\in \mathbb{C}^n, \|v\|=1\},
\end{equation}
which extends the numerical range from the 2D complex plane to the 3D space, with the additional $z$-axis reflecting the matrix gains. Fig.~\ref{Fig_DW_shell_SRG} shows an example DW shell $DW(A)$, whose projection onto the $x$-$y$ plane gives the numerical range of $A$. Its top and bottom points along the $z$-axis correspond to the squares of the largest and smallest singular values of $A$, respectively. Although DW-shell separation between interconnected transfer matrices certifies feedback stability~\cite{huang2025geometric}, this full 3D separation is not directly amenable to decentralized verification, motivating lower-dimensional projection-based sufficient tests. In~\cite{huang2025geometric}, the projection of the DW shell onto the
$x$-$z$ plane is defined as the \textbf{$\bm x$-$\bm z$ graph}, which is
\begin{equation}\label{eq:xz_graph}
    P_{xz}(A)= \{(\Re(v^*Av),\|Av\|^2) : v \in \mathbb{C}^n, \|v\|=1\} \,,
\end{equation}
where we use $\Re(\cdot)$ to denote the real part of a complex number. We refer to~\cite{zhang2025phantom} and~\cite{huang2025geometric} for the properties of DW shells.

\begin{figure}[!t]
	\centering
	\includegraphics[width=3.0in]{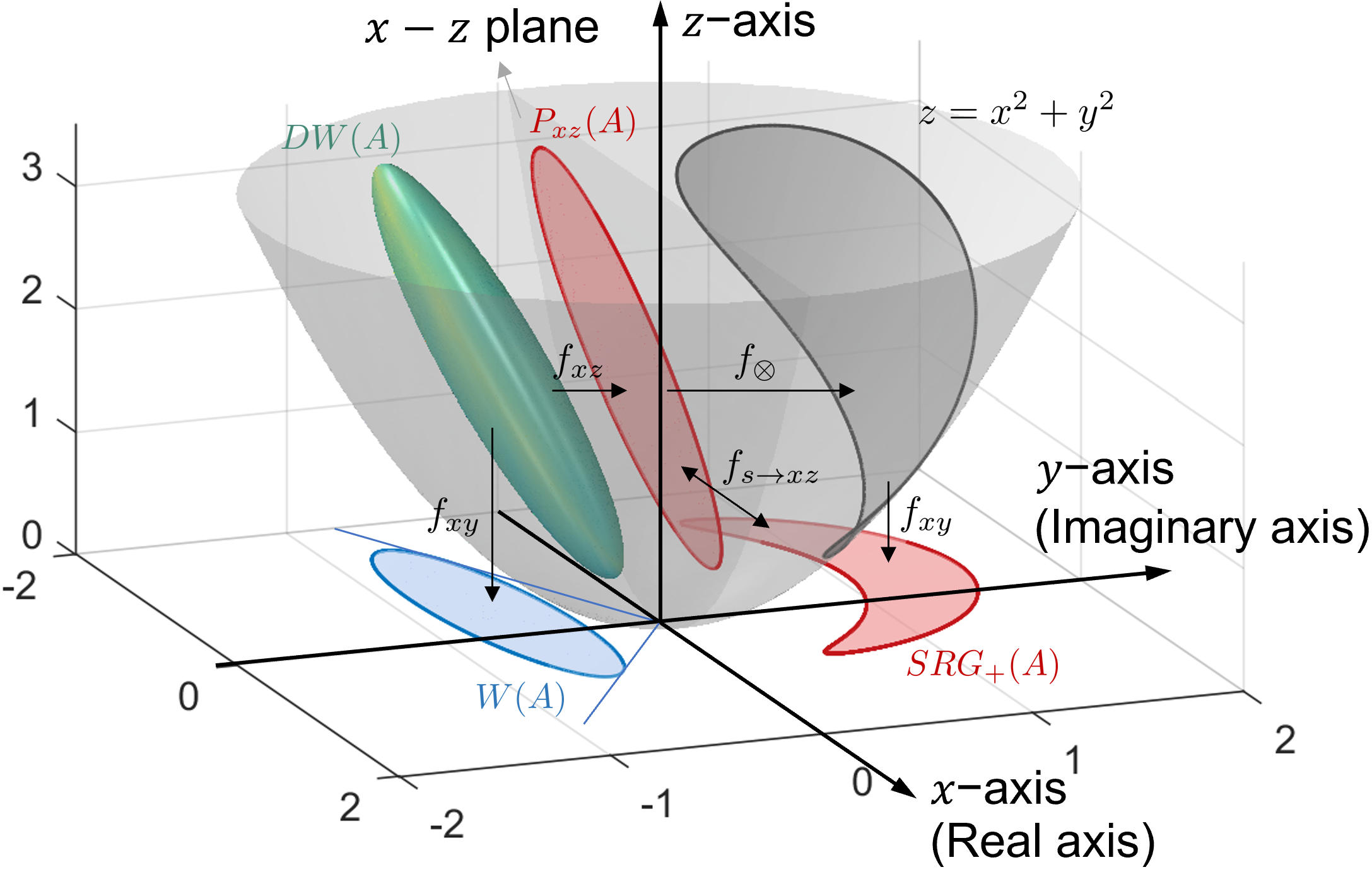}
	\vspace{-3mm}
	%\DeclareGraphicsExtensions.
	\caption{Illustration of the DW shell $DW(A)$, numerical range $W(A)$, $x$-$z$ graph $P_{xz}(A)$, and the positive part of the scaled relative graph $SRG_+(A)$. The $W(A)$ and $P_{xz}(A)$ are projections of $DW(A)$ onto the $x$-$y$ and $x$-$z$ planes via $f_{xy}$ and $f_{xz}$, respectively. The $SRG_+(A)$ and $P_{xz}(A)$ are in bijection via $f_{s\to xz}$, visualized by mapping $P_{xz}(A)$ to the shaded gray patch on the paraboloid $z=x^2+y^2$ via $f_{\otimes}$ and then to the $x$-$y$ plane via $f_{xy}$.}
	\vspace{-4mm}
	\label{Fig_DW_shell_SRG}
    \end{figure}

The concept of scaled relative graphs (SRGs) has also recently gained attention in analyzing dynamical systems~\cite{chaffey2023graphical, baron2025decentralized}, as it simultaneously combines the gain and phase information of an operator. For a complex matrix $A \in \mathbb{C}^{n \times n}$, the \textbf{SRG} can be written in the normalized form as
\begin{equation}\label{eq:SRG2}
\begin{split}
 \!\!\!\! SRG(A) =& \left\{ \|Av\|e^{\pm j{\rm arccos}\left( \frac{\Re(v^*Av)}{\|Av\|}\right)}\!: v\!\in\! \mathbb{C}^n, \|v\|\!=\!1 \right\}\\ 
     =& \left\{  
    \begin{array}{ll}
       c \in \mathbb{C}:\ & \Re(c) = \Re(v^*Av), |c| = \|Av\|,  \\
       &v\in \mathbb{C}^n, \|v\|=1
    \end{array}
    \right\} \,.
\end{split}
\end{equation}

The SRG is defined on the complex plane and is symmetric about the real axis. Moreover, by combining~\eqref{eq:xz_graph} and~\eqref{eq:SRG2} we find that the positive part (above the real axis) of the SRG of a complex matrix $A$, denoted by $SRG_+(A)$, can be uniquely mapped to the $x$-$z$ graph of $A$ through a bijective function $f_{s \rightarrow xz}:\mathbb{C} \rightarrow \mathbb{R}\times \mathbb{R}$ as
\begin{equation}\label{eq:bijective}
    f_{s \rightarrow xz}(c) = (x,z) = (\Re(c),|c|^2) \,.
\end{equation}
Hence, the SRG and $x$-$z$ graph have a one-to-one relationship, and both are related to the DW shell, as also investigated in~\cite{zhang2025phantom}. It can also be seen from Fig.~\ref{Fig_DW_shell_SRG} that $SRG_+(A)$ can be obtained by projecting $P_{xz}(A)$ onto the paraboloid $z = x^2+y^2$ and then onto the $x$-$y$ plane.
However, the SRG and $x$-$z$ graph have different geometric shapes and properties, e.g., the $x$-$z$ graph is always convex~\cite{huang2025geometric} while the SRG is not. This connection between SRG and $x$-$z$ graph will be used next to formulate a more general decentralized stability certificate.

\subsection{Geometric Decentralized Stability Certificate}

As discussed in Section~\ref{2}, the closed-loop dynamics of the PE-dominated power system shown in Fig.~\ref{Fig_closed_loopYY} exhibit two structural properties: $\widetilde{\mathbf Y}_{\mathrm C}^{N}(s,\mathcal O)$ is block-diagonal, and $\widetilde{\mathbf Y}_{\rm grid}$ is constant and positive definite. For brevity of notation, let $H= \widetilde{\mathbf Y}_{\rm grid}^{-1} \succ 0$ and $G(s,\mathcal O)=e^{J{\boldsymbol\theta}}\,\widetilde{\mathbf Y}_{\mathrm C}^{N}(s,\mathcal{O})\, e^{-J{\boldsymbol\theta}}$, which admits a block-diagonal structure as
\begin{equation}
\label{diag_structure}
G(s,\mathcal O)=\operatorname{diag}\{G_1(s,\mathcal O_1),\;\ldots,\; G_N(s,\mathcal O_N)\},
\end{equation}
where $G_i(s,\mathcal O_i)$ captures the dynamics of the $i$-th converter at its operating point $\mathcal O_i$. We further notice that $e^{J\theta_i}$ is a unitary matrix and will not affect the shape of DW shells or $x$-$z$ graphs. Hence, we drop $e^{J{\boldsymbol\theta}}$ and let $G(s,\mathcal O)=\widetilde{\mathbf Y}_{\mathrm C}^{N}(s,\mathcal{O})$, so $G_i(s,\mathcal O_i)=\widetilde{\mathbf Y}_{{\mathrm C},i}(s,\mathcal{O}_i)$.
These structural properties, illustrated in Fig.~\ref{Fig_closed_loop}, enable the following geometric decentralized stability conditions, which extend the framework in~\cite{huang2025geometric} by further incorporating the SRG separation condition and accounting for the impact of variable operating points.

\begin{theorem}[Decentralized stability certificate under variable operating points]
\label{thm:Decentralized}
Let \(G,H\in\mathcal{RH}_{\infty}^{m\times m}\), where
\(G(s,\mathcal O)\) is block-diagonal as in~\eqref{diag_structure} and \(H\succ0\) is a constant positive-definite matrix. Let \(\mathcal S_i\subseteq\bar{\mathcal S}_i\) be a nonempty subset of the \(i\)-th admissible operating-point set. The closed-loop system \(G(s,\mathcal O)\#H\) is stable for all \(\mathcal O\in\mathcal S_1\times\mathcal S_2\times\cdots\times\mathcal S_N\) if, for each \(\mathcal O_i\in\mathcal S_i\) and for each \(\omega\in[0,\infty)\), \textbf{either}
\begin{enumerate}[1)]
    \item \label{theGain}
    the decentralized \textbf{gain} condition, i.e.,
    \begin{equation}\label{eq:DecGain}
    \max_i \, \overline{\sigma}(G_i(j\omega,\mathcal O_i)) 
    < \underline{\sigma}(H^{-1}), \quad holds, \textbf{or}
    \end{equation}
    
    \item \label{thePhase}
    the decentralized \textbf{phase} condition, i.e.,
    \begin{equation}\label{eq:DecPhase}
    \begin{cases}
    \text{a) } 
    \max\limits_i \, \overline{\phi}(G_i(j\omega,\mathcal O_i)) < \pi, \\[2pt]
    \text{b) } 
    \min\limits_i \, \underline{\phi}(G_i(j\omega,\mathcal O_i)) > -\pi, \\[2pt]
    \text{and c) } 
    \max\limits_i \, \overline{\phi}(G_i(j\omega,\mathcal O_i))
    -
    \min\limits_i \, \underline{\phi}(G_i(j\omega,\mathcal O_i))
    < \pi,
    \end{cases}
    \end{equation}
    holds, \textbf{or}
    
    \item \label{theXZ}
    for each \(i\), the decentralized \textbf{\(\bm{x}\)-\(\bm{z}\) graph separation}, i.e.,
    \begin{equation}\label{eq:DecXZ}
    P_{xz}(G_i(j\omega,\mathcal O_i))
    \cap
    P_{xz}\!\left(-\tfrac{1}{\tau}H^{-1}\right)
    =
    \varnothing,
    \end{equation}
    holds \(\forall \tau \in (0,1]\), \textbf{or equivalent to 3)},
    
    \item \label{theSRG}
    for each \(i\), the decentralized \textbf{SRG separation}, i.e.,
    \begin{equation}\label{eq:DecSRG}
    SRG(G_i(j\omega,\mathcal O_i))
    \cap
    SRG\!\left(-\tfrac{1}{\tau}H^{-1}\right)
    =
    \varnothing,
    \end{equation}
    holds \(\forall \tau \in (0,1]\).
\end{enumerate}
\end{theorem}
\noindent The parameter \(\tau\in(0,1]\) is the scaling parameter used in the sufficient generalized Nyquist criterion adopted from~\cite{huang2025geometric}.

\begin{proof}

    Consider any operating point \(\mathcal O = [\mathcal O_1^\top, \dots, \mathcal O_N^\top]^\top\) where \(\mathcal O_i\in\mathcal S_i\) for all \(i\). For this fixed \(\mathcal O\), the theorem requires that for each \(\omega\in[0,\infty)\), one of the four decentralized conditions (gain, phase, \(x\)-\(z\) graph, or SRG) holds. The gain and phase conditions in \eqref{eq:DecGain} and \eqref{eq:DecPhase} have been proposed and verified in \cite{huang2024gain}. The \(x\)-\(z\) graph separation condition in \eqref{eq:DecXZ} has been further proved in \cite{huang2025geometric}, which is complementary to the gain and phase conditions since they all imply the separation of DW shells. Moreover, under the bijective function in \eqref{eq:bijective}, \(P_{xz}(G_i(j\omega,\mathcal O_i))\) (or \(P_{xz}(-\frac{1}{\tau}H^{-1})\)) can be uniquely mapped to the positive part of \(SRG(G_i(j\omega,\mathcal O_i))\) (or \(SRG(-\frac{1}{\tau}H^{-1})\)), and vice versa, so Condition~4 holds if and only if Condition~3 holds. In short, for any \(\mathcal O\in\mathcal S_1\times\mathcal S_2\times\cdots\times\mathcal S_N\), the condition of DW shell separation holds if any of the four conditions is satisfied, and the closed-loop system \(G(s,\mathcal O)\#H\) is stable~\cite{huang2025geometric}. This completes the proof.
\end{proof}

The conditions in Theorem~\ref{thm:Decentralized} differ in conservativeness and computational burden, making them complementary for certifying subsets \(\mathcal S_i\) of admissible operating points. The gain/phase conditions are simpler to check, while the \(x\)-\(z\) graph and SRG separations retain more geometric information and can reduce conservativeness. Following the idea of mixed geometric stability conditions, the subsets \(\mathcal S_i\) are certified as stability regions when, for each \(\mathcal O_i\in\mathcal S_i\) and each frequency, at least one of these decentralized conditions is satisfied.

It can be seen from conditions~\eqref{eq:DecGain}-\eqref{eq:DecSRG} that the overall system stability can be analyzed in a modular and decentralized way. Specifically, for each converter \(i\), these conditions are checked for all \(\mathcal O_i\in\mathcal S_i\), rather than by enumerating all system operating-point combinations. Once the subsets \(\mathcal S_i\) are certified as stability regions, Theorem~\ref{thm:Decentralized} guarantees the stability of the multi-converter system for any \(\mathcal O\in\mathcal S_1\times\mathcal S_2\times\cdots\times\mathcal S_N\). The reason why we can split \(\mathcal O\) into the individual \(\mathcal O_i\) is that the system has a special structure as in Fig.~\ref{Fig_closed_loop}, where \(\mathcal O_i\) enters only into the block \(G_i(s,\mathcal O_i)\) and \(\mathcal O\) does not affect \(H\) at all.

\begin{figure}[!t]
    \vspace{0mm}
	\centering
	\includegraphics[width=2.5in]{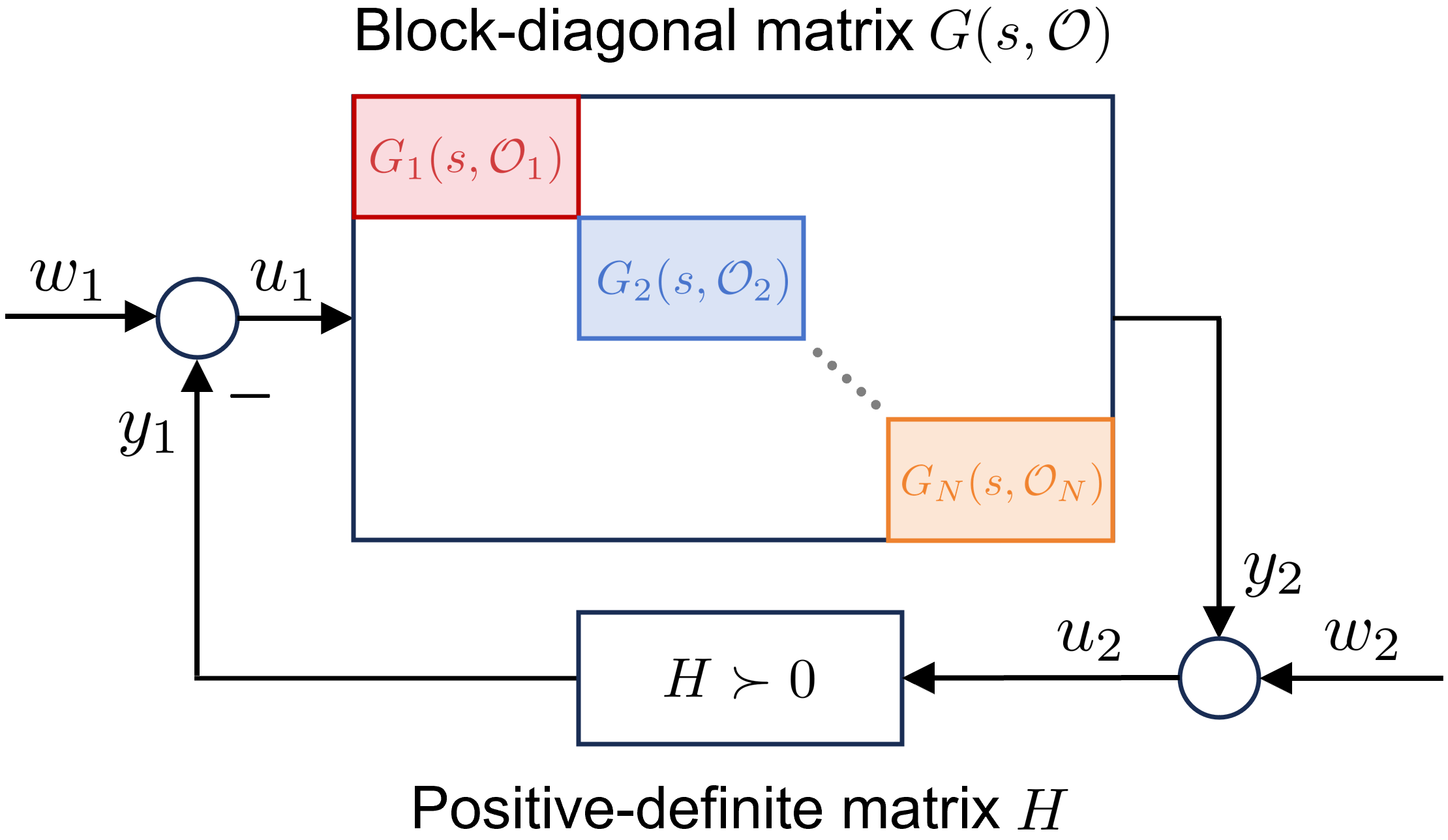}
	\vspace{-2mm}
	%\DeclareGraphicsExtensions.
	\caption{A closed-loop system \(G(s,\mathcal O)\#H\), where \( G(s,\mathcal O) \) is block-diagonal and \( H \) is a constant, positive-definite matrix.}
	\vspace{-4mm}
	\label{Fig_closed_loop}
    \end{figure}

%\subsection{Decentralized Stability Operating Region (DSOR)}

\section{Decentralized Stability Operating Region \\ and Its Efficient Construction} \label{4}
In this section, we define a geometric stability margin and the corresponding decentralized stability operating region using \(x\)-\(z\) graphs. 
We then develop a fast algorithm to construct this region, guided by worst-case power-injection analysis.

\subsection{Decentralized Stability Operating Region}
We are now interested in finding the largest subsets \(\mathcal S_i\) certifiable by the conditions in Theorem~\ref{thm:Decentralized}. 
To this end, we focus on the \(x\)-\(z\) graph condition, which is generally the least conservative among them, and the resulting subsets thus yield the largest certified stability region for \(\mathcal O\). As a first step, we define the decentralized stability margin for each converter.

\begin{definition}[Decentralized stability margin based on $x$-$z$ graphs]
\label{def:margin}
For the $i$-th converter with operating point $\mathcal O_i$, the decentralized stability margin is defined as
\begin{equation}
\label{eq:MCi_def}
\!\!M_{i}(\mathcal O_i)\! := 
    \operatorname{dist}\!\left\{
    \bigcup_{\scriptscriptstyle\omega\in[0,\infty)} \hspace{-2mm} P_{xz}\bigl(G_i(j\omega,\mathcal O_i)\bigr),
      \bigcup_{\scriptscriptstyle\tau\in(0,1]} \hspace{-1mm} P_{xz}\bigl(-\tfrac{1}{\tau}H^{-1}\bigr)\!
    \right\},
\end{equation}
where $\operatorname{dist}\{A,B\}:=\min_{x\in A,\,y\in B}\|x-y\|$ is the Euclidean distance between two sets.
\end{definition}

We note that the \(x\)-\(z\) graph and SRG separation conditions are equivalent due to the bijection in~\eqref{eq:bijective}. However, unlike SRGs, which are usually nonconvex, \(x\)-\(z\) graphs are always convex. For the considered \(G_i(j\omega,\mathcal O_i)\), which is a \(2\times2\) complex matrix at any \(\omega\), its \(x\)-\(z\) graph is a solid ellipse~\cite{huang2025geometric}. This convexity makes the distance computation more convenient, and thus we define the stability margin based on \(x\)-\(z\) graphs.

    \begin{figure}[!t]
    \vspace{-1.5mm}
	\centering
	\includegraphics[width=2.7in]{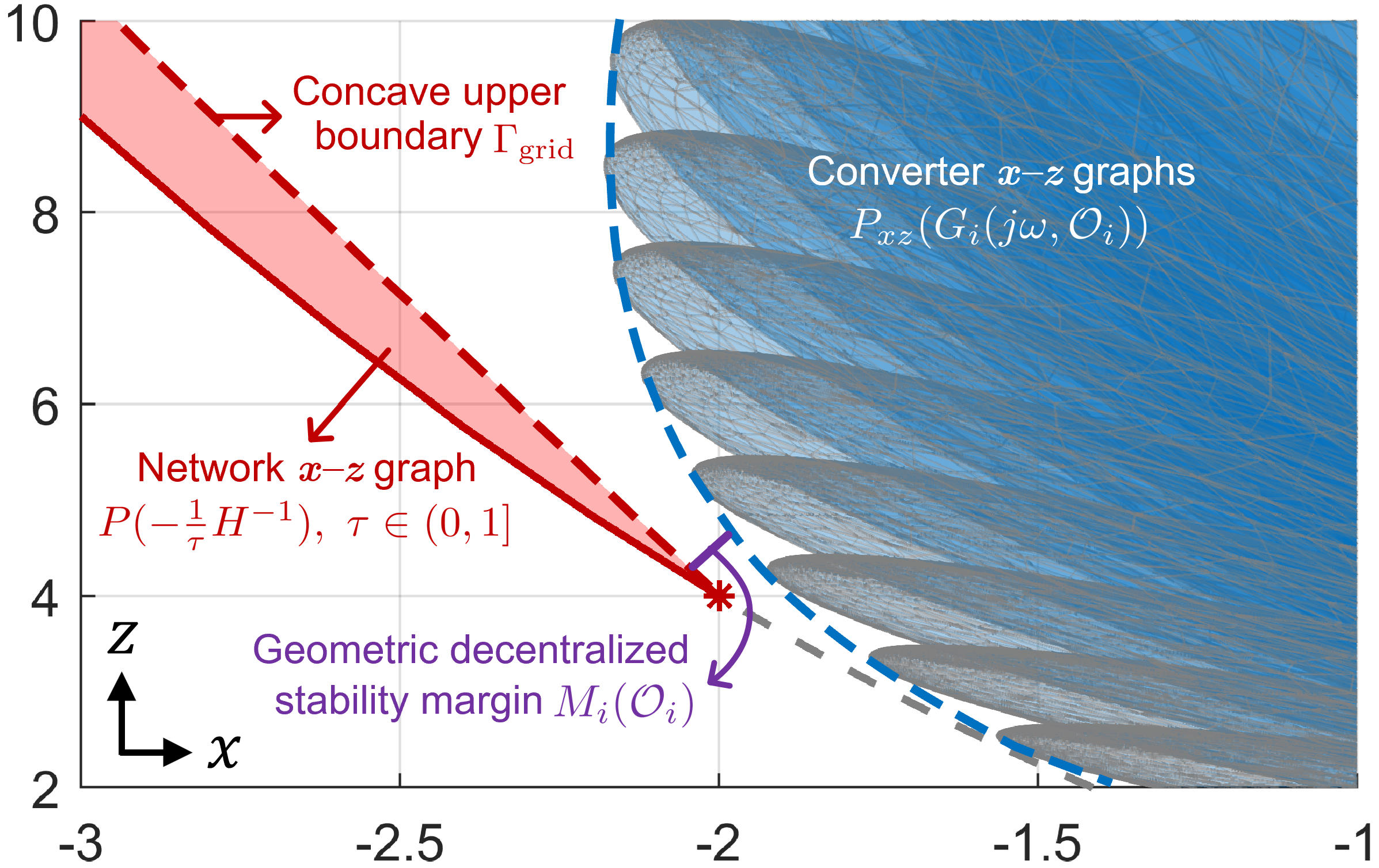}
	\vspace{-2mm}
	%\DeclareGraphicsExtensions.
	\caption{Illustration of the stability margin \(M_i(\mathcal O_i)\), i.e., the distance between the converter's \(x\)-\(z\) graphs \(P_{xz}(G_i(j\omega,\mathcal O_i))\) and the network's \(x\)-\(z\) graph trajectory \(P_{xz}(-\frac{1}{\tau}H^{-1})\) over \(\tau\in(0,1]\), with upper boundary \(\Gamma_{\rm grid}\).}
	\vspace{-5mm}
	\label{Fig_geometric_margin}
    \end{figure}

\begin{remark}
Since $\widetilde{\mathbf Y}_{\rm grid}\succ0$, the $x$-$z$ graph $P_{xz}(-\tfrac{1}{\tau}\widetilde{\mathbf Y}_{\rm grid})$ over $\tau\in(0,1]$ has an upper boundary $\Gamma_{\rm grid}$ that is concave upward~\cite{huang2025geometric}, as shown in Fig.~\ref{Fig_geometric_margin}. We can use this upper boundary to simplify the computation of stability margin as
\begin{equation}
\label{eq:MCi_equiv}
M_{i}(\mathcal O_i) = 
    \operatorname{dist}\left\{
        \bigcup_{\scriptscriptstyle\omega\in[0,\infty)} \hspace{-2mm} P_{xz}\bigl(G_i(j\omega,\mathcal O_i)\bigr),\;
        \Gamma_{\rm grid}
    \right\}.
\end{equation}
\end{remark}

It can be seen from Theorem~\ref{thm:Decentralized} and the definition of \(M_i(\mathcal O_i)\) that a positive margin, i.e., \(M_i(\mathcal O_i) > 0\), implies \(x\)-\(z\) graph separation for the $i$-th converter. Hence, if \(M_i(\mathcal O_i) > 0\) holds for all the converters, then the closed-loop system is stable at \(\mathcal O = [\mathcal{O}_1^\top,\ldots,\mathcal{O}_N^\top]^\top\). Accordingly, the largest $\mathcal S_i$ that satisfies the condition in~\eqref{eq:DecXZ}, referred to as the decentralized stability operating region (DSOR), can be formulated as
\begin{equation}
\label{eq:DSOR_def}
\mathcal D_i := \bigl\{ \mathcal O_i \in \bar{\mathcal S}_i \;\big|\; M_i(\mathcal O_i) > 0 \bigr\}.
\end{equation}
Further, the complementary region \(\bar{\mathcal S}_i \setminus \mathcal D_i\) is referred to as the risky operating region, where the decentralized stability condition is violated and instability may arise.
% Since the \(x\)-\(z\) graph condition is less conservative than the gain and phase conditions, \(\mathcal D_i\) defined via \(M_i(\mathcal O_i) > 0\) yields the largest certified stability subset of \(\bar{\mathcal S}_i\) among those obtainable from the conditions in Theorem~\ref{thm:Decentralized}.
Based on~\eqref{eq:DSOR_def}, we can obtain the following stability condition, which evaluates the system stability by examining the operating points.
\begin{corollary}[Decentralized stability certificate based on DSORs]
\label{cor:DSOR}
Consider  \(G,H\in\mathcal{RH}_{\infty}^{m\times m}\), where
\(G(s,\mathcal O)\) is block-diagonal as in~\eqref{diag_structure}
and \(H\succ0\) is a constant positive-definite matrix. The closed-loop system \(G(s,\mathcal O)\#H\), which represents the multi-converter system in Fig.~\ref{common_topology}, is stable if \(\mathcal O_i \in \mathcal D_i\) for all $i \in \{1,...,N\}$, or equivalently, \(\mathcal O\in \mathcal D_1 \times \cdots \times \mathcal D_N\).
\end{corollary}

\begin{proof}
By definition, $\mathcal O_i \in \mathcal D_i$ implies $M_i(\mathcal O_i) > 0$, so the condition of $x$-$z$ graph separation holds for the $i$-th converter. If this holds for all $i \in \{1,...,N\}$, the condition in~\eqref{eq:DecXZ} is satisfied. Hence, the closed-loop system is stable according to Theorem~\ref{thm:Decentralized}. This completes the proof.
\end{proof}

Corollary~\ref{cor:DSOR} enables a direct decentralized verification procedure focusing on the converter's operating points. Specifically, once $\mathcal D_i$ is obtained for each converter, the stability of a system operating point $\mathcal O = [\mathcal O_1^\top,\dots,\mathcal O_N^\top]^\top$ is certified by checking if $\mathcal O_i \in \mathcal D_i$ for all $i$. Hence, stability analysis over $\mathcal O$ in the $3N$-dimensional space $\bar{\mathcal S}$ reduces to $N$ independent tests of $\mathcal O_i$, avoiding the problem of combinatorial explosion.

The key remaining issue is how to efficiently construct the DSOR \(\mathcal D_i\) for each converter. Since \(\mathcal D_i\) is defined based on the stability margin \(M_i(\mathcal O_i)\), its construction requires extensive evaluations of \(M_i(\mathcal O_i)\) over \(\mathcal O_i \in \bar{\mathcal S}_i\). Each evaluation involves repeated distance computations between the upper boundary $\Gamma_{\rm grid}$ and the converter's $x$-$z$ graphs at different frequencies, as shown in~\eqref{eq:MCi_equiv}. To make this procedure tractable, the next subsection develops a continuation-based method to achieve fast stability margin evaluation and DSOR construction.

\begin{algorithm}[t]
\caption{Efficient Evaluation of $M_i(\mathcal O_i)$ in~\eqref{eq:MCi_equiv}}
\label{alg:fast_margin}
\footnotesize

\textbf{Input:} operating point $\mathcal O_i$; grid boundary $g:[x_L,x_R]\rightarrow\Gamma_{\mathrm{grid}}$; frequency range $[f_{\min},f_{\max}]$; step size of frequency scan $\Delta f$; numbers of candidate frequency points $N_l,N_g$; tolerances $\varepsilon,\varepsilon_{\mathrm{op}}$; operating point cache $\mathcal C_{\mathrm{op}}$.

\vspace{2pt}
\hrule
\vspace{2pt}

\noindent\textbf{Subroutine: distance evaluation at a given frequency $d_i(f,\mathcal O_i)$}

\vspace{-2pt}

\begin{enumerate}[(a)]

\item Set $\omega=2\pi f$ for $\widetilde{\mathbf Y}_{{\mathrm C},i}(j\omega,\mathcal O_i)$, and then parameterize the converter's $x$-$z$ graph as $p(\xi)=c+U\,[a\cos\xi,\;b\sin\xi]^\top$, where $\xi\in[0,2\pi)$.

\item Solve $d_i(f,\mathcal O_i)=\min_{\xi\in[0,2\pi),\,x\in[x_L,x_R]}\|p(\xi)-g(x)\|_2$ by damped Newton iterations to obtain $(\xi^\star,x^\star)$. If the frequency cache $\mathcal C_{\mathrm f}\neq\varnothing$, set the initial $(\xi,x)$ to $(\xi_c^\star,x_c^\star)$ from the nearest cached frequency $f_c$; otherwise, initialize $(\xi,x)$ by uniform sampling over $[0,2\pi)\times[x_L,x_R]$.

\item Store $(f,\xi^\star,x^\star)$ in $\mathcal C_{\mathrm f}$ and return $d_i(f,\mathcal O_i)$.

\end{enumerate}

\vspace{2pt}
\hrule
\vspace{2pt}

\noindent\textbf{Main routine: continuation-based margin evaluation}

\vspace{-2pt}

\begin{enumerate}[1)]

\item If there exists a cached pair $(\mathcal O_c,f_c^\star)$ in $\mathcal C_{\mathrm{op}}$ with $\|\mathcal O_i-\mathcal O_c\|<\varepsilon_{\mathrm{op}}$, choose $N_l$ candidate frequency points $\{f_k\}$ uniformly from $[\max(f_{\min},f_c^\star-\Delta f),\;\min(f_{\max},f_c^\star+\Delta f)]$; otherwise choose $N_g$ candidate frequency points $\{f_k\}$ uniformly from $[f_{\min},f_{\max}]$.

\item For each $f_k$, evaluate $d_k = d_i(f_k,\mathcal O_i)$, and let $\hat{k} = \operatorname{arg\,min}_k d_k$. 

\item Form a local search interval $[f_L,f_R]$ using the adjacent coarse points around $f_{\hat{k}}$, and solve $f^\star = \operatorname{arg\,min}_{f\in[f_L,f_R]} d_i(f,\mathcal O_i)$ by a one-dimensional search (e.g., \texttt{fminbnd}). Set $M_i(\mathcal O_i) = d_i(f^\star,\mathcal O_i)$.

\item If \(M_i(\mathcal O_i)>\varepsilon\), store
\((\mathcal O_i,f^\star)\) in \(\mathcal C_{\mathrm{op}}\);
otherwise mark as zero-margin.

\end{enumerate}

\textbf{Output:} Geometric decentralized stability margin $M_{i}(\mathcal O_i)$ and frequency $f^\star$.

\end{algorithm}

\subsection{Efficient Construction of DSORs}
\label{sec:fast_computation}

According to~\eqref{eq:MCi_equiv}, evaluating $M_i(\mathcal O_i)$ requires finding the minimum distance between the converter's $x$-$z$ graphs and the grid boundary $\Gamma_{\rm grid}$ over frequency. In practical computation, the search is performed over a prescribed frequency range $[f_{\min},f_{\max}]$. At a certain frequency $\renewcommand{\arraystretch}{0.7}f={\omega}/{2\pi}$, the boundary of the converter's $x$-$z$ graph is an ellipse and thus can be parameterized by $p(\xi)$ where $\xi\in[0,2\pi)$. Hence, the distance evaluation at this frequency is formulated as the closest-point search between $p(\xi)$ and the grid boundary $\Gamma_{\rm grid}$ parameterized by $g(x)$ with $x\in[x_L,x_R]$, expressed as $d_i(f,\mathcal O_i)=\min_{\xi\in[0,2\pi),\,x\in[x_L,x_R]}\|p(\xi)-g(x)\|_2$. Then, an outer frequency search over $f \in [f_{\min},f_{\max}]$ determines the minimizing frequency $f^\star$ and gives $M_i(\mathcal O_i)=d_i(f^\star,\mathcal O_i)$. Algorithm~\ref{alg:fast_margin} summarizes this evaluation procedure.

Moreover, Algorithm~\ref{alg:fast_margin} accelerates repeated margin evaluations using continuation-based initialization at both the operating-point and frequency-scan levels. At the operating-point level, the cache $\mathcal C_{\rm op}$ stores entries $(\mathcal O_c,f_c^\star)$, where $f_c^\star$ is the minimizing frequency at a previously evaluated operating point $\mathcal O_c$. For a nearby $\mathcal O_i$, the coarse frequency scan is thus restricted to a local window centered at $f_c^\star$ instead of the full frequency range. At the frequency level, the cache $\mathcal C_{\rm f}$ stores entries $(f_c,\xi_c^\star,x_c^\star)$, where $(\xi_c^\star,x_c^\star)$ initializes the damped Newton iterations at the neighboring frequencies. The algorithm then uses a coarse frequency scan to locate a candidate interval and refines $f^\star$ by a one-dimensional search.

Although Algorithm~\ref{alg:fast_margin} evaluates $M_i(\mathcal O_i)$ at a single operating point $\mathcal O_i$, it can be repeatedly called over sampled operating points in $\bar{\mathcal S}_i$ to construct $\mathcal D_i$. In this process, continuation-based initialization reduces the cost of repeated evaluations, and $\mathcal D_i$ is obtained by identifying the operating points satisfying $M_i(\mathcal O_i)>0$. In what follows, a single-converter example is provided to illustrate the DSOR construction and the computational efficiency of the proposed algorithm.

\begin{example}[Single-converter DSOR construction]\label{ex:single_converter}
Consider a single-converter system where a GFL converter operates in either PQ or PV control mode. The control structure is shown in Fig.~\ref{common_topology}, and the system parameters are provided in the Supplementary Material. The admissible operating set is specified by 
\(\kappa=1.2\), \(\Psi=[-\pi,\pi)\), and \(\mathcal V=[0.9,1.1]\,\mathrm{p.u.}\)

We apply Algorithm~\ref{alg:fast_margin} to evaluate \(M_i(\mathcal O_i)\) over the sampled admissible operating set. Fig.~\ref{Fig_DSOR_DROR}~(a) and (c) show the resulting stability margin distributions and the corresponding DSORs under PQ and PV control modes, respectively. In this example, the risky regions mainly occur at high active-power injection with low voltage in PQ mode, and high active-power absorption with high voltage in PV mode. For comparison, a brute-force baseline is implemented by uniformly discretizing \(f\), \(\xi\), and \(x\), without using continuation at either the operating-point or frequency level. On an AMD Ryzen 7 PRO 6850H processor with \(32\)~GB RAM, the brute-force baseline requires \(3.414\,\mathrm{s}\) per operating point on average, whereas Algorithm~\ref{alg:fast_margin} requires only \(0.018\,\mathrm{s}\), corresponding to a \(189.7\times\) speedup.
\end{example}

The example above verifies the effectiveness and computational efficiency of Algorithm~\ref{alg:fast_margin} in constructing DSORs. Building on this procedure, the next step is to further speed up the operating-point scan by identifying where small stability margins are more likely to occur. We observe that a small margin is usually reflected by the left-hand side of the converter's \(x\)-\(z\) graphs approaching the grid boundary $\Gamma_{\rm grid}$. The next subsection further analyzes this phenomenon to provide insight into the efficient search for the DSOR boundary.

     \begin{figure}[!t]
     \vspace{0mm}
    	\centering
    	\includegraphics[width=0.95\linewidth]{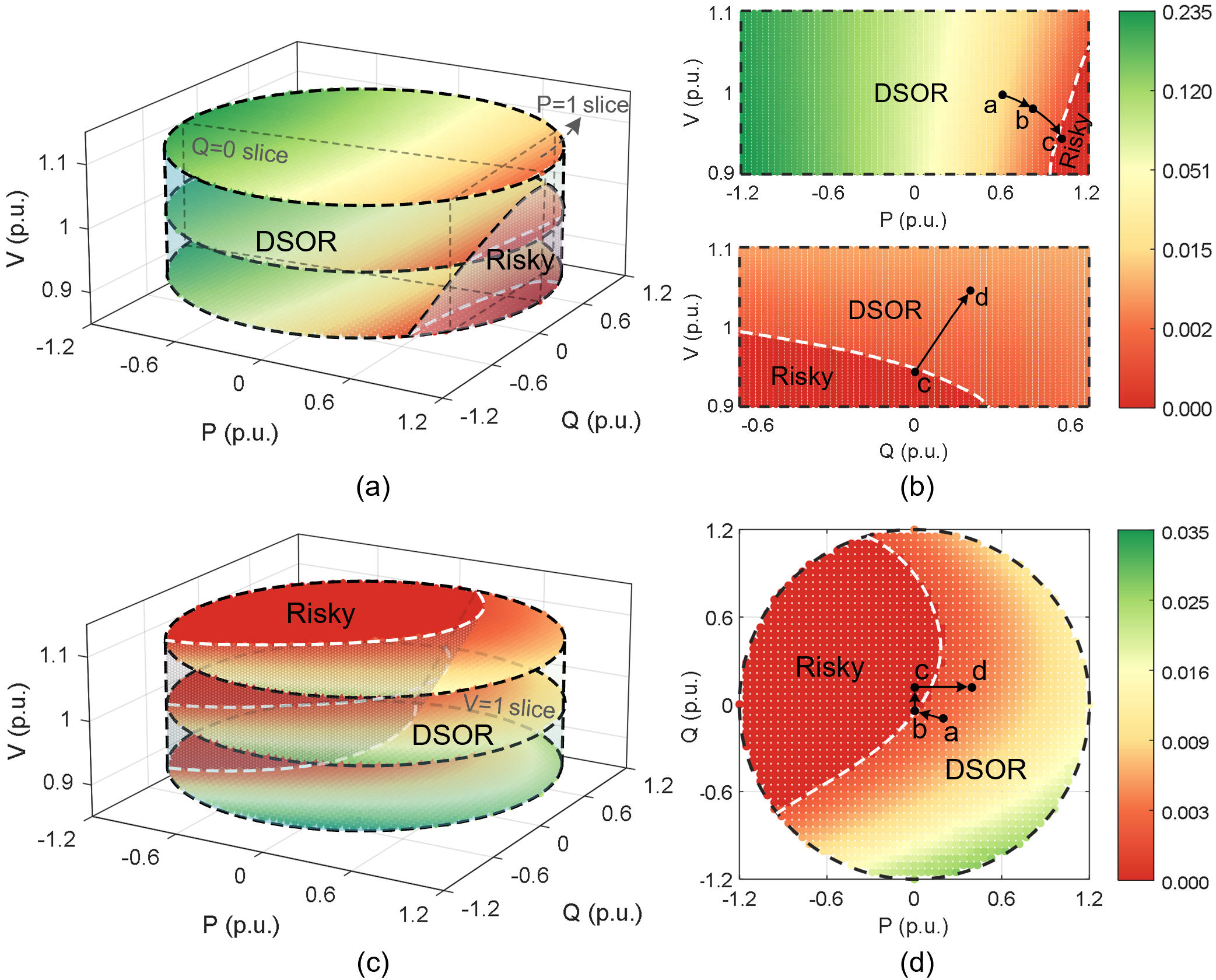}
    	\vspace{-2mm}
    	%\DeclareGraphicsExtensions.
    	\caption{Decentralized operating regions of a single-converter system. (a) PQ mode. (b) Vertical slices of the PQ operating domain at \(Q=0\) and \(P=1\). (c) PV mode. (d) Horizontal slice of the PV operating domain at \(V=1\).}
    	\vspace{-4mm}
    	\label{Fig_DSOR_DROR}
    \end{figure}

\subsection{Worst Case Analysis of Power Injections}\label{worst case}
We next investigate how to find the worst-case operating conditions with respect to \(M_i(\mathcal O_i)\) and further alleviate the computational burden of DSORs. It can be seen from Fig.~\ref{Fig_geometric_margin} that if the converter's $x$-$z$ graphs move leftward, then \(M_i(\mathcal O_i)\) decreases and the system may become unstable. Hence, 
we analyze how the left bounds of the converter's \(x\)-\(z\) graphs change with the power injections \((P_i,Q_i)\). Specifically, for a frequency \(\omega\in[0,\infty)\) and an operating point \(\mathcal O_i\), the leftmost \(x\)-value of the converter's \(x\)-\(z\) graph defined in \eqref{eq:xz_graph} is
\begin{equation}
\begin{split}
\underline{x}_i(\omega,\mathcal O_i)= & 
\min_{(x,z)\in P_{xz}(\widetilde{\mathbf Y}_{{\rm C},i}(j\omega,\mathcal O_i))} x \\
=&
\min_{\|v\|_2=1}
\Re\!\left(
v^*
\mathbf Y_{{\rm C},i}(j\omega,\mathcal O_i)
\mathbf F_\epsilon^{-1}(j\omega)
v
\right) \\
=&
\min_{\|v\|_2=1}
v^*\mathbf H_i(\omega,\mathcal O_i)v ,
\end{split}
\label{eq:leftmost_x_H}
\end{equation}
where the Hermitian matrix \(\mathbf H_i(\omega,\mathcal O_i)\) is defined as
\begin{equation}
\!\! \mathbf H_i(\omega,\mathcal O_i)
:=
\tfrac{
\mathbf Y_{{\rm C},i}(j\omega,\mathcal O_i)\mathbf F_\epsilon^{-1}(j\omega)
+
\left[
\mathbf Y_{{\rm C},i}(j\omega,\mathcal O_i)\mathbf F_\epsilon^{-1}(j\omega)
\right]^*
}{2}.
\label{eq:H_i_def}
\end{equation}

According to~\eqref{eq:Y0_PQ_PV}, for fixed \(V_i\),
\(\mathbf Y_{{\rm C},i}(j\omega,\mathcal O_i)\) is affine in \((P_i,Q_i)\)
under both PQ and PV control modes. Since \(\mathbf F_\epsilon^{-1}(j\omega)\) is independent of the operating point and taking the Hermitian part is a linear operation, we can deduce that \(\mathbf H_i(\omega,\mathcal O_i)\) is also an affine Hermitian matrix in \((P_i,Q_i)\), expressed as
\begin{equation}
\mathbf H_i(\omega,\mathcal O_i)
=
\mathbf H_{i,0}(\omega,V_i)
+
P_i\mathbf H_{i,P}(\omega,V_i)
+
Q_i\mathbf H_{i,Q}(\omega,V_i),
\label{eq:affine_Hermitian}
\end{equation}
where \(\mathbf H_{i,0}(\omega,V_i)\), \(\mathbf H_{i,P}(\omega,V_i)\), and \(\mathbf H_{i,Q}(\omega,V_i)\) are Hermitian matrices independent of \((P_i,Q_i)\). Based on this affine property, we obtain the following result.

\begin{proposition}[Concavity of the function \(\underline{x}_i(\omega,\mathcal O_i)\) and properties of its minimum]
\label{prop:monotonicity}
For any fixed \(V_i\) and
\(\omega\in[0,\infty)\), the function
\(\underline{x}_i(\omega,\mathcal O_i)\) in~\eqref{eq:leftmost_x_H} is concave with respect to \((P_i,Q_i)\). Hence, for any nonempty compact power-injection region $\mathcal U_i$ with $(\mathcal U_i,V_i)\subset  \bar{\mathcal S}_i$, the minimum of \(\underline{x}_i\) over \(\mathcal U_i\) is attained on its boundary, that is,
\begin{equation}
\min_{(P_i,Q_i)\in\mathcal U_i}
\underline{x}_i(\omega,\mathcal O_i)
=
\min_{(P_i,Q_i)\in\partial\mathcal U_i}
\underline{x}_i(\omega,\mathcal O_i),
\label{eq:boundary_general}
\end{equation}
where $\partial\mathcal U_i$ denotes the boundary of $\mathcal U_i$. 
For some nonempty compact regions
\(\{\mathcal U_i(\rho)\}_{\rho\ge0}\) satisfying
\(\mathcal U_i(\rho_1)\subseteq\mathcal U_i(\rho_2)\) for
\(0\le\rho_1<\rho_2\), the boundary minimum is non-increasing, i.e.,
\begin{equation}
\min_{(P_i,Q_i)\in\partial\mathcal U_i(\rho_2)}
\underline{x}_i(\omega,\mathcal O_i)
\le
\min_{(P_i,Q_i)\in\partial\mathcal U_i(\rho_1)}
\underline{x}_i(\omega,\mathcal O_i).
\label{eq:boundary_monotonicity}
\end{equation}
\end{proposition}

\begin{proof}
The proof is given in the Supplementary Material.
\end{proof}

     \begin{figure}[!t]
     \vspace{-3mm}
    	\centering
    	\includegraphics[width=0.7\linewidth]{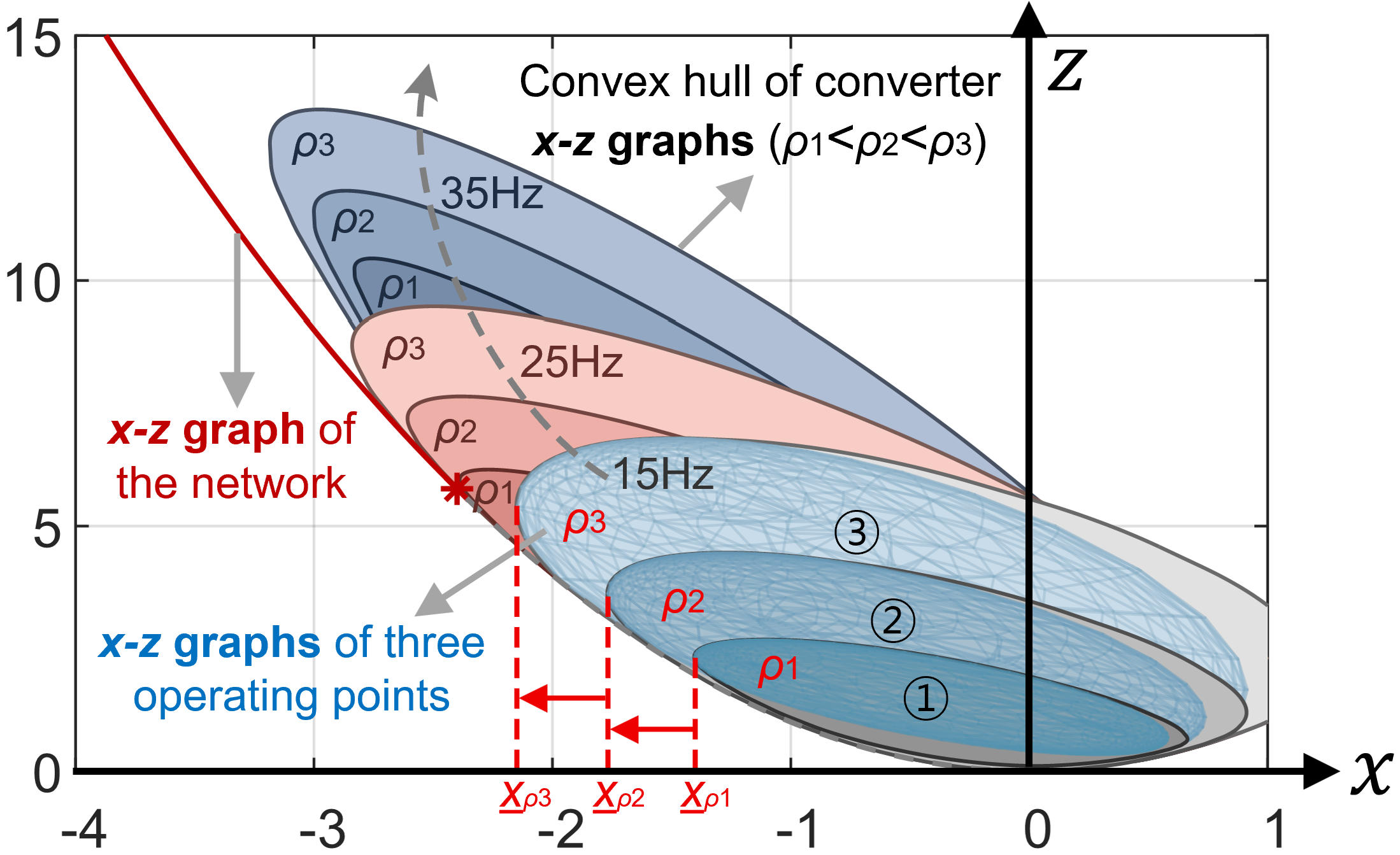}
    	\vspace{-2mm}
    	%\DeclareGraphicsExtensions.
        \caption{Illustration of the convex hulls of the converter's $x$-$z$ graphs under different $\mathcal U_i(\rho)$ with \(\rho_1<\rho_2<\rho_3\). The convex hulls are generated by 16 uniformly sampled points in $\Psi_i\in[0,2\pi)$, each paired with 5 samples of $V_i\in[0.9,1.1]\,\mathrm{p.u.}$ Three marked operating points $(P_i,Q_i,V_i)$ are \ding{172}~$(0.28,0.28,0.90)$, \ding{173} $(0.57,0.57,0.90)$, and \ding{174} $(0.85,0.85,0.90)$.}
    	\vspace{-4mm}
    	\label{Fig_PQPV_xz}
    \end{figure}

Proposition~\ref{prop:monotonicity} shows how the leftmost $x$-value of the converter's $x$-$z$ graph changes with the power injections $(P_i,Q_i)$. To be specific, the worst-case scenario occurs on the boundary of the power variation set where $(P_i,Q_i) \in \mathcal U_i$. Moreover, with a larger set, the left bound moves towards the left, indicating a smaller stability margin \(M_i(\mathcal O_i)\) or even instability, as illustrated in Fig.~\ref{Fig_PQPV_xz} where we use the same converter setting (PQ mode) as in Example~\ref{ex:single_converter}.
At a certain frequency, the convex hulls of the converter's \(x\)-\(z\) graphs are generated from uniformly sampled operating points on \(P^2+Q^2=\rho_k^2\) with \(\rho_1<\rho_2<\rho_3\). With the increase of \(\rho\), the left bound of the convex hull shifts
towards the left at each frequency, as reflected by \(\underline{x}_{\rho_3}<\underline{x}_{\rho_2}<\underline{x}_{\rho_1}\). This moves the converter's \(x\)-\(z\) graphs towards the network's \(x\)-\(z\) graph, resulting in a smaller \(M_i(\mathcal O_i)\). Hence, the larger-\(\rho\) boundary layers are more critical when searching for the worst-case operating points. For instance, if we know that each point on the boundary of a power variation set is stable, then all the sets that are contained in this set are stable, and there is no need to examine them. This motivates a radial boundary-to-interior scan. To illustrate the idea, for each \(0\le\rho\le\kappa_i\), we define
\begin{equation}
\begin{aligned}
\mathcal U_i(\rho)
&=
\bigl\{
\mathbf u_i\!=\!(r\cos\psi,r\sin\psi)\!:
0\le r\le \rho,\ \! \psi\in\!\Psi_i
\bigr\},\\
\mathcal O_i^{\partial}(\rho)
&=
\bigl\{
[\mathbf u_i,\ V_i]^\top:
\mathbf u_i\in\partial\mathcal U_i(\rho),\
V_i\in\mathcal V_i
\bigr\},
\end{aligned}
\label{eq:radial_boundary_operating_set}
\end{equation}
where \(\mathbf u_i=(P_i,Q_i)\) denotes the power-injection point,
\(\Psi_i=[-\pi,\pi)\) gives the full disk, and
\(\Psi_i=[0,\pi/2]\) corresponds to \(P_i,Q_i\ge0\).
The admissible voltage set is \(\mathcal V_i=[V_{i,\min},V_{i,\max}]\) in PQ mode and
\(\mathcal V_i=\{V_{{\rm ref},i}\}\) in PV mode.

The practical scan then proceeds from the outer boundary to the interior to efficiently find the DSOR boundary. Algorithm~\ref{alg:fast_margin} is first applied to the outermost boundary operating-point set \(\mathcal O_i^{\partial}(\kappa_i)\). If operating points with \(M_i(\mathcal O_i)=0\) are identified on this boundary, the scan is continued by evaluating the complete boundary sets \(\mathcal O_i^{\partial}(\rho)\) for a sequence of decreasing \(\rho\). In this way, the DSOR boundary can be found efficiently.

%========================

\section{Case Studies} \label{5}
This section validates our method by simulating a single-converter system and a real-world wind power system that contains 54 wind farms (each one modeled by a converter). 
% The single-converter case examines the DSOR characterization via geometric illustrations and time-domain simulations. The wind power system demonstrates the method's applicability to multi-converter systems with variable operating points and its capability to guide control modification.

\vspace{-2mm}

 \subsection{Case Studies of a Single-Converter System}
 
Based on the DSORs constructed in Example~\ref{ex:single_converter}, this subsection validates their effectiveness through $x$-$z$ graphs and SRGs at different operating points as well as time-domain simulations under changes of operating points.

Fig.~\ref{Fig_PQ_xz_SRG} plots the \(x\)-\(z\) graphs of the converter (PQ mode) and the network at two representative operating points \(a\) and \(c\) in Fig.~\ref{Fig_DSOR_DROR}~(b), along with their \(SRG_+\). At operating point \(a\), the converter's \(x\)-\(z\) graphs remain separated from the network's \(x\)-\(z\) graph in Fig.~\ref{Fig_PQ_xz_SRG}~(a), indicating a positive margin, so \(a\) lies inside the DSOR. Due to the bijective relationship between the \(x\)-\(z\) graph and \(SRG_+\), the corresponding \(SRG_+\) plots in Fig.~\ref{Fig_PQ_xz_SRG}~(b) also exhibit separation. At operating point \(c\), the converter and network \(x\)-\(z\) graphs intersect in Fig.~\ref{Fig_PQ_xz_SRG}~(c), so \(M_i(c)=0\) and \(c\) lies outside the DSOR. The corresponding \(SRG_+\) plots in Fig.~\ref{Fig_PQ_xz_SRG}~(d) also intersect. These plots illustrate the convexity of the \(x\)-\(z\) graphs, which makes them easier to evaluate than the generally nonconvex \(SRG_+\). Fig.~\ref{Fig_PQ_xz_SRG} focuses on a certain frequency range, and at the other frequencies, either the gain condition~\eqref{eq:DecGain} or the phase condition~\eqref{eq:DecPhase} is satisfied, similar to~\cite{huang2025geometric}, and thus the plots are omitted here.

     \begin{figure}[!t]
     \vspace{-3mm}
    	\centering
    	\includegraphics[width=0.98\linewidth]{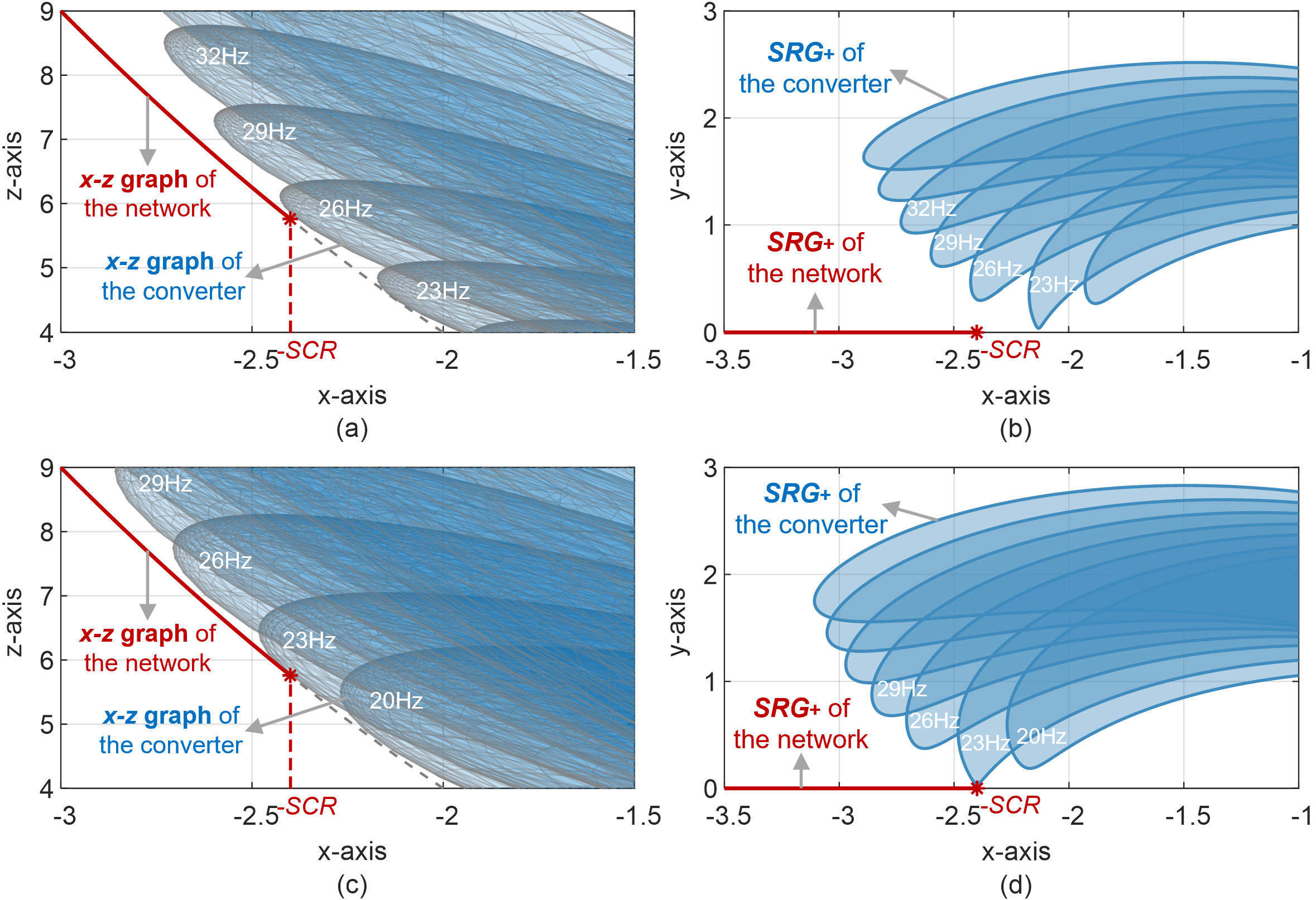}
    	\vspace{-2mm}
    	%\DeclareGraphicsExtensions.
        \caption{\(x\)-\(z\) graphs and \(SRG_+\) of the converter and the network at operating points \(a\) and \(c\) in Fig.~\ref{Fig_DSOR_DROR}(b) under PQ mode. (a) \(x\)-\(z\) graph and (b) \(SRG_+\) at operating point \(a\); (c) \(x\)-\(z\) graph and (d) \(SRG_+\) at operating point \(c\).}
    	\vspace{-2mm}
    	\label{Fig_PQ_xz_SRG}
    \end{figure}

     \begin{figure}[!t]
     \vspace{0mm}
    	\centering
    	\includegraphics[width=0.98\linewidth]{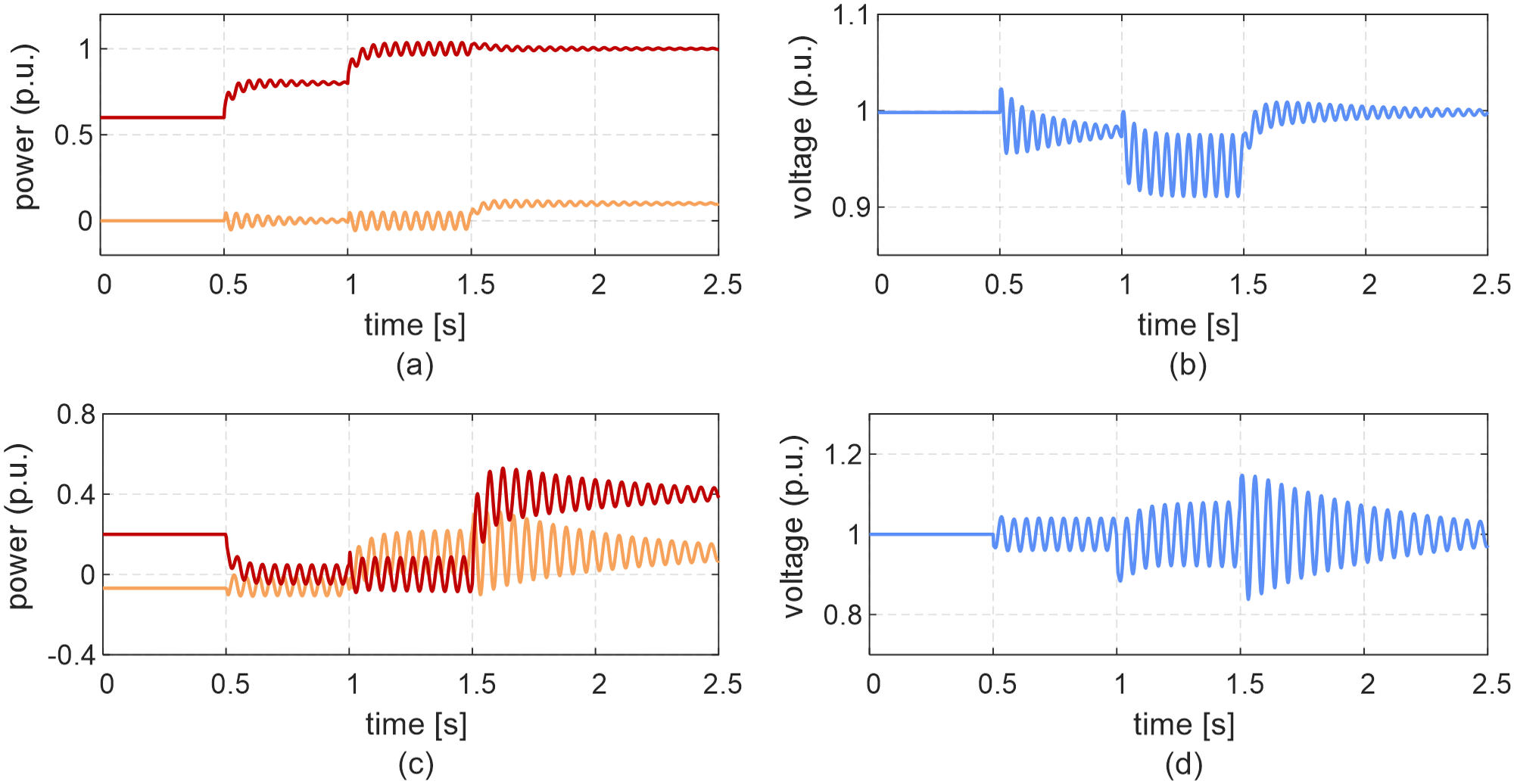}
    	\vspace{-2mm}
    	%\DeclareGraphicsExtensions.
        \caption{Time-domain responses of power and voltage of the single-converter system under PQ and PV control modes, with (a)-(b) under PQ mode and (c)-(d) under PV mode. \textcolor[HTML]{C00000}{\textbf{---}}: active power; \textcolor[HTML]{F7A25B}{\textbf{---}}: reactive power; \textcolor[HTML]{5B8FF7}{\textbf{---}}: voltage.}
    	\vspace{-4mm}
    	\label{Fig_PQPV_time_domain}
    \end{figure}

We next consider the changes of operating points, as shown in Fig.~\ref{Fig_DSOR_DROR}, and present the time-domain simulation results. Under PQ mode, the initial operating point is \(P=0.6\,\mathrm{p.u.}\) and \(Q=0\,\mathrm{p.u.}\), with the voltage \(V=0.998\,\mathrm{p.u.}\) determined by the power flow. As the active power increases to \(0.8\,\mathrm{p.u.}\) and then to \(1.0\,\mathrm{p.u.}\), the operating point moves from \(a\) to \(b\) and then to \(c\) in Fig.~\ref{Fig_DSOR_DROR}~(b), which leaves the DSOR and enters the risky operating region. By fixing \(P=1.0\,\mathrm{p.u.}\) and increasing the reactive power to \(0.2\,\mathrm{p.u.}\), the operating point further moves to \(d\), where it returns to the DSOR. Under PV mode, the initial operating point is \(P=0.2\,\mathrm{p.u.}\) and \(V=1.0\,\mathrm{p.u.}\), with \(Q=-0.068\,\mathrm{p.u.}\) determined by the power flow. As the active power decreases to \(0\,\mathrm{p.u.}\), the operating point moves outside the DSOR at point \(b\), shown in Fig.~\ref{Fig_DSOR_DROR}~(d). A voltage dip at the infinite bus further drives the operating point to \(c\), while increasing the active power to \(0.4\,\mathrm{p.u.}\) drives it to \(d\), where it returns to the DSOR. 
The corresponding time-domain plots of both PQ and PV cases are given in Fig.~\ref{Fig_PQPV_time_domain}, where the system shows sustained oscillations when the operating point leaves the DSOR, verifying the effectiveness of the DSOR analysis.

\subsection{Case Studies of a Wind Power System}

The effectiveness of the proposed method is further verified on the wind power system shown in Fig.~\ref{hami_topology}, which has five regions (A-E), including 54 wind farms operating in PQ mode. Each wind farm is modeled by a GFL converter whose control structure is shown in Fig.~\ref{common_topology}. The PLL bandwidths of the converters are set to \(108~\mathrm{rad/s}\) in Regions A and D, and \(28~\mathrm{rad/s}\) in Regions B, C, and E. All converters operate with zero reactive power setpoints, while the active power setpoints of the 54 converters are provided in the Supplementary Material. The detailed control and network parameters are also provided in the Supplementary Material.

    \begin{figure}[!tb]
     \vspace{-3mm}
    \centerline{\includegraphics[width=0.9\linewidth]{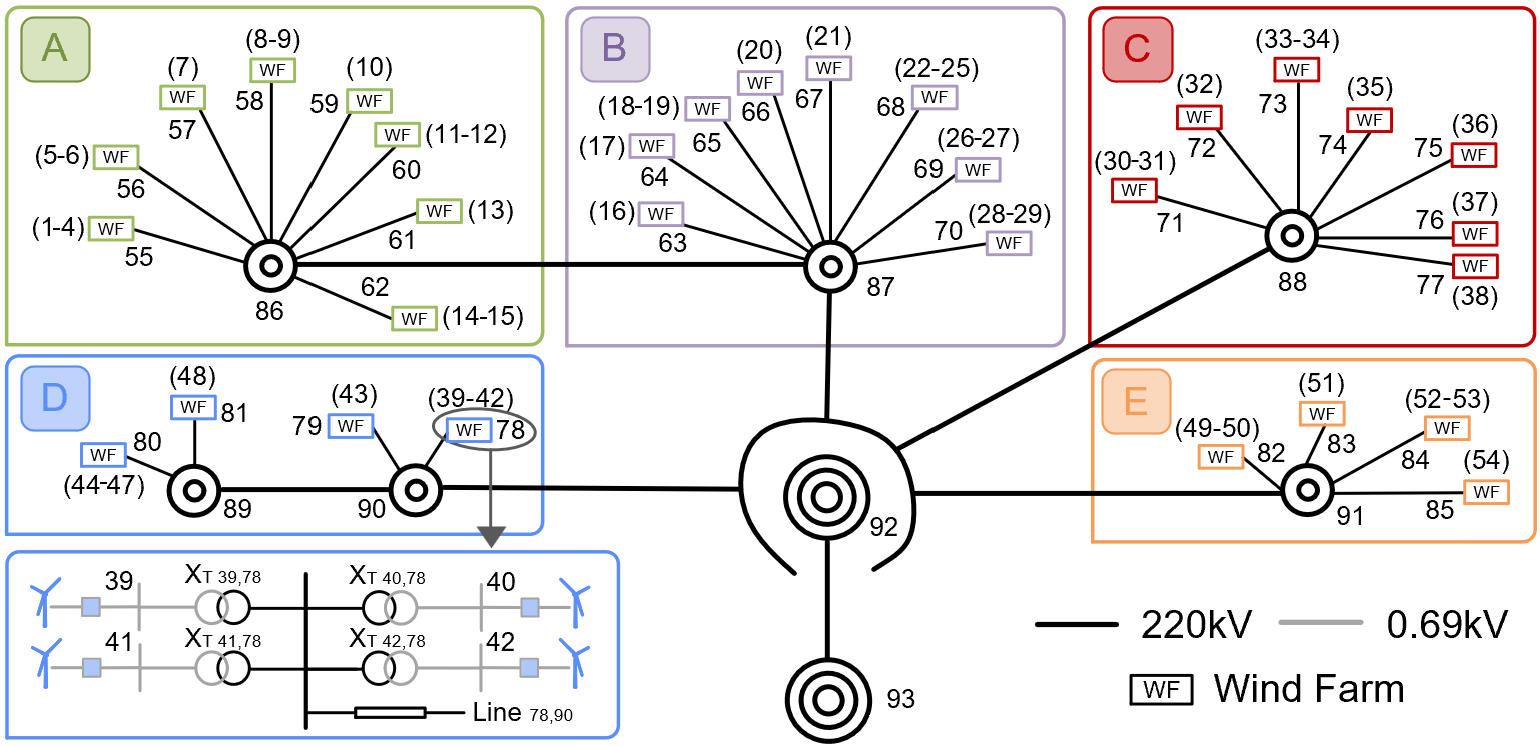}}
     \vspace{-2mm}
     \caption{Topology of a real-world wind power system in China.}
     \label{hami_topology}
     \vspace{-4mm}
    \end{figure}

The wind farms can be divided into two groups according to the PLL bandwidths: Group~1 includes Regions A and D, and Group~2 includes Regions B, C, and E. The converters in the same group can still have different operating points. Following the boundary-to-interior scan described in Section~\ref{worst case}, the convex hulls of converter \(x\)-\(z\) graphs are constructed on the outer boundary set \(\mathcal O_i^{\partial}(\kappa_i)\) in~\eqref{eq:radial_boundary_operating_set} with \(\kappa_i=1\). Since each wind farm operates in the \(P_i,Q_i\ge0\) quarter-sector, the power-injection boundary is uniformly sampled by \(45\) points in \(\Psi_i=[0,\pi/2]\), each paired with \(5\) voltage samples in \(V_i\in[0.9,1.1]\,\mathrm{p.u.}\) Fig.~\ref{hami_convex_xz}~(a) and (b) demonstrate the resulting convex hulls of the converters' \(x\)-\(z\) graphs, where the convex hulls of Group~2 remain separated from the network's \(x\)-\(z\) graph, whereas those of Group~1 intersect with the network's \(x\)-\(z\) graph. The zoom-in view in Fig.~\ref{hami_convex_xz}~(c) further shows that the intersection occurs at \(P=1.0\,\mathrm{p.u.}\) Fig.~\ref{hami_convex_xz}(e) and (f) show the corresponding \(SRG_+\) plots, which exhibit the same separation/intersection due to their bijective relationship with the \(x\)-\(z\) graphs. This intersection indicates that the separation condition in~\eqref{eq:DecXZ} is violated for Group~1, revealing a potential instability risk at some operating points.

To eliminate this risk, the PLL bandwidths of the converters in Group~1 are reduced to \(58~\mathrm{rad/s}\). After this modification, the convex hulls of the \(x\)-\(z\) graphs for Group~1 become separated from the network's \(x\)-\(z\) graphs, as shown in Fig.~\ref{hami_convex_xz}~(d). Together with the separation observed for Group~2, all the converters satisfy the decentralized \(x\)-\(z\) graph separation condition over their admissible operating regions. Hence, by Corollary~\ref{cor:DSOR}, the wind power system is certified to be stable for any admissible system operating point \(\mathcal O\).

\begin{figure}[!t]
    \vspace{-3mm}
	\centerline{\includegraphics[width=0.99\linewidth]{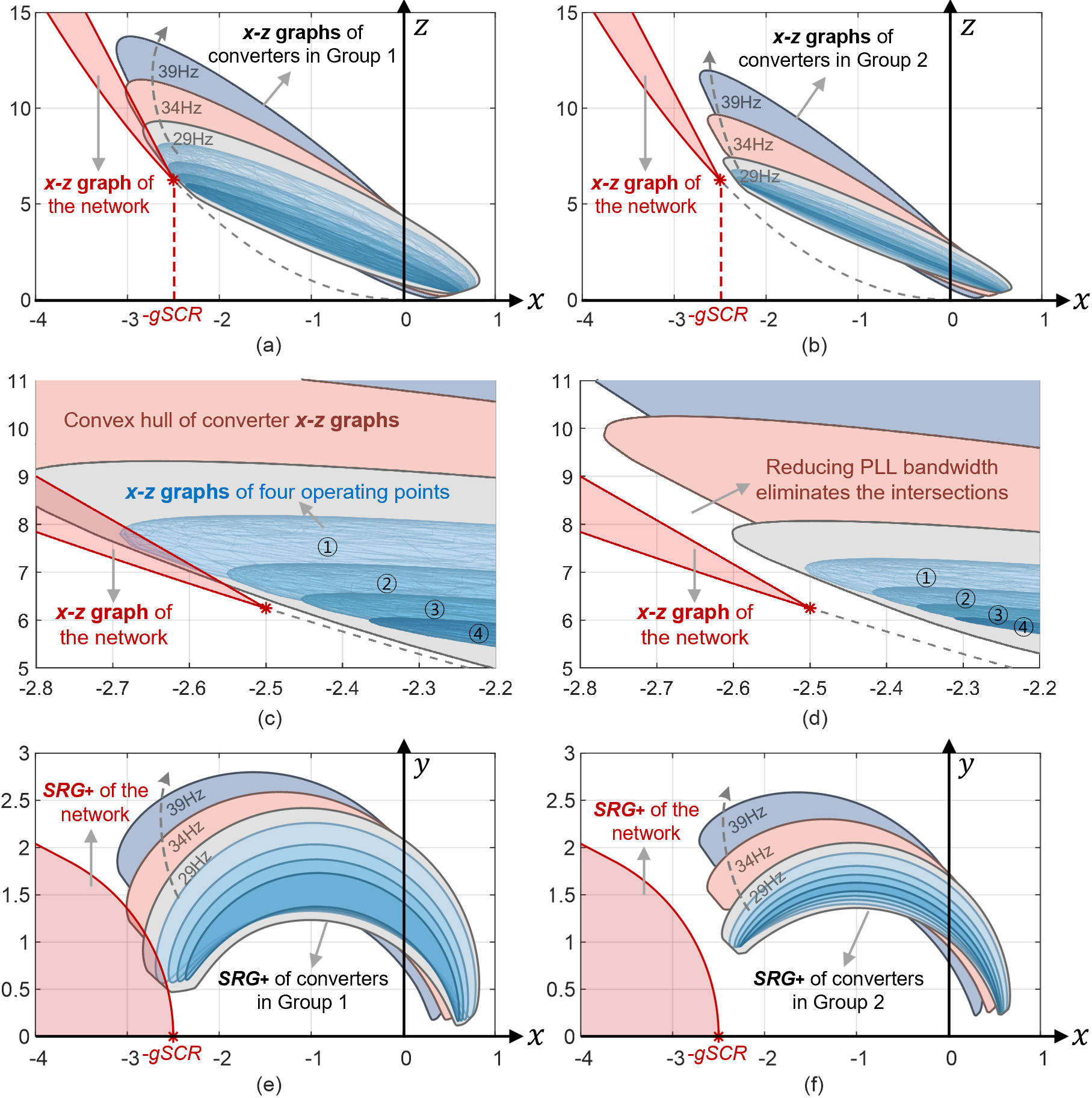}}
	\vspace{-2mm}
	%\DeclareGraphicsExtensions.
    \caption{Convex hulls of the converters' $x$-$z$ graphs; the network's $x$-$z$ graphs; and their corresponding SRGs. (a) \(x\)-\(z\) graphs for Group~1; (b) \(x\)-\(z\) graphs for Group~2; (c) The zoom-in view of (a); (d) The \(x\)-\(z\) graphs for Group~1 after PLL modification; (e) The \(SRG_+\) corresponding to (a); (f) The \(SRG_+\) corresponding to (b). Four example operating points \((P_i,Q_i,V_i)\) are: \ding{172} $(1.0,0,1.0)$, \ding{173} $(0.6,0,1.0)$, \ding{174} $(0.3,0,1.0)$, and \ding{175} $(0,0,1.0)$.}
	\vspace{-4mm}
	\label{hami_convex_xz}
\end{figure}

Time-domain simulations are then conducted to validate the above analysis. Initially, all the converters in Group~1 operate with \(P\leq0.6\,\mathrm{p.u.}\), corresponding to the stable operating points shown in Fig.~\ref{hami_convex_xz}~(c). Then, the active power of Converters 44-47 in Group~1 is increased from \(0.6\) to \(1.0\,\mathrm{p.u.}\) at \(t=0.2\,\mathrm{s}\), followed by a power increase of Converters 1-4 from \(0.5\) to \(0.9\,\mathrm{p.u.}\) at \(t=1\,\mathrm{s}\). Before the control modification, these changes result in the intersection in Fig.~\ref{hami_convex_xz}~(c). Fig.~\ref{hami_time_domain} shows that the original system exhibits sustained oscillations after \(0.2\,\mathrm{s}\) and instability after \(1\,\mathrm{s}\), while the modified system remains stable under the same operating-point changes.

Finally, to further show the decentralized stability margins, the frequency-wise distance \(d_i(f,\mathcal O_i)\) in Algorithm~\ref{alg:fast_margin} is evaluated within \(0\)-\(100~\mathrm{Hz}\) for each converter after the two power steps, as shown in Fig.~\ref{frequency_wise_distance}. Before the control modification, \(d_i(f,\mathcal O_i)\) reaches zero for Converters 1-4 and 44-47 within \(25\)-\(35~\mathrm{Hz}\), indicating that their decentralized stability margins are zero at these operating points. After the modification, we have \(d_i(f,\mathcal O_i)>0\) and \(M_i(\mathcal O_i)>0\) for all the converters. This again confirms the \(x\)-\(z\) graph separation after the modification and is consistent with the time-domain results.

    \begin{figure}[!tb]
     \vspace{0mm}
    \centerline{\includegraphics[width=0.99\linewidth]{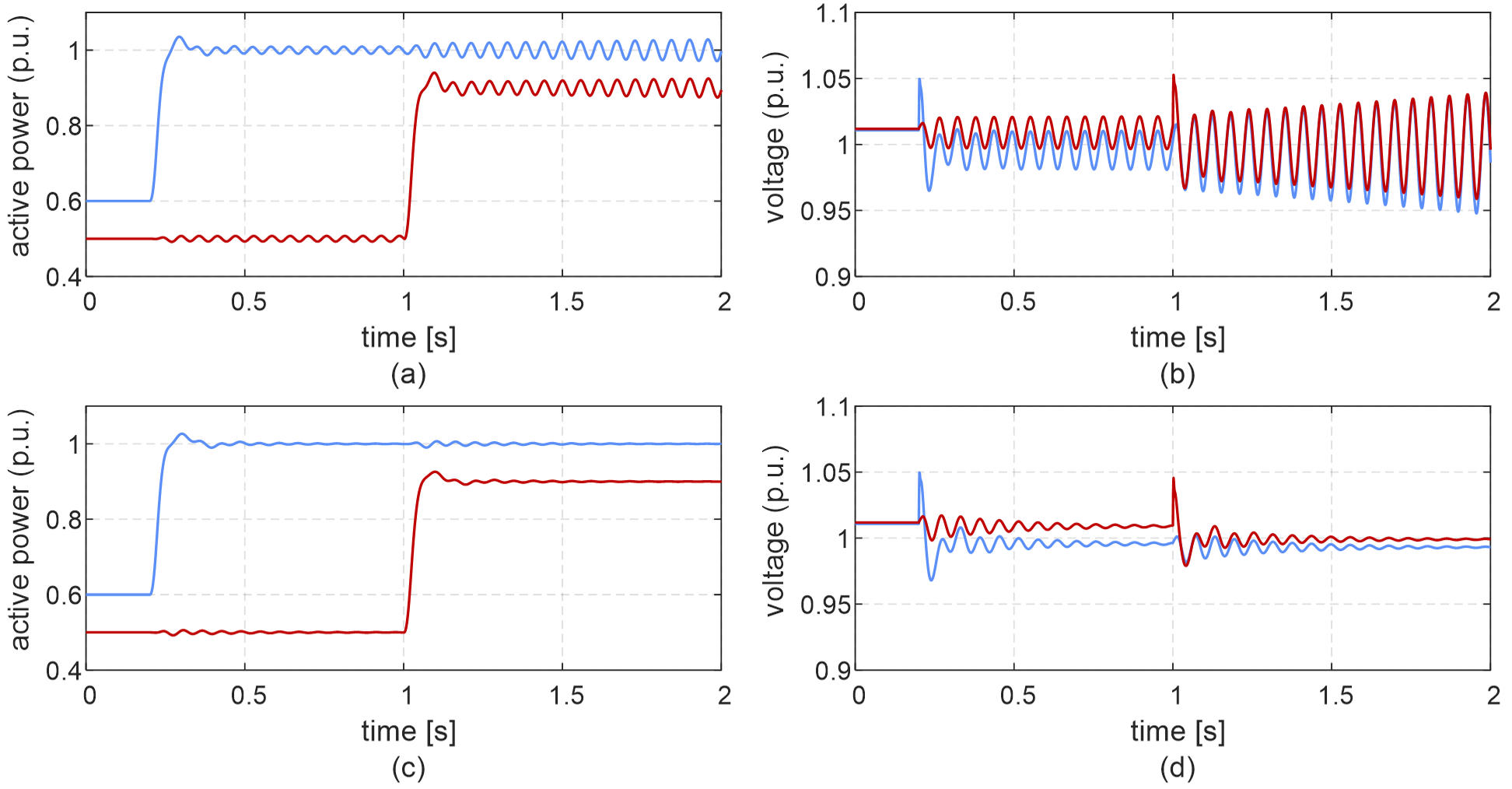}}
     \vspace{-3mm}
     \caption{Time-domain responses of active power and voltage of Converters 1 and 44 in the wind power system, with (a)-(b) before PLL modification and (c)-(d) after PLL modification. \textcolor[HTML]{C00000}{\textbf{---}}: Converter~1; \textcolor[HTML]{5B8FF7}{\textbf{---}}: Converter~44.}
     \label{hami_time_domain}
     \vspace{-2mm}
    \end{figure}

    \begin{figure}[!t]
    \vspace{0mm}
	\centerline{\includegraphics[width=0.99\linewidth]{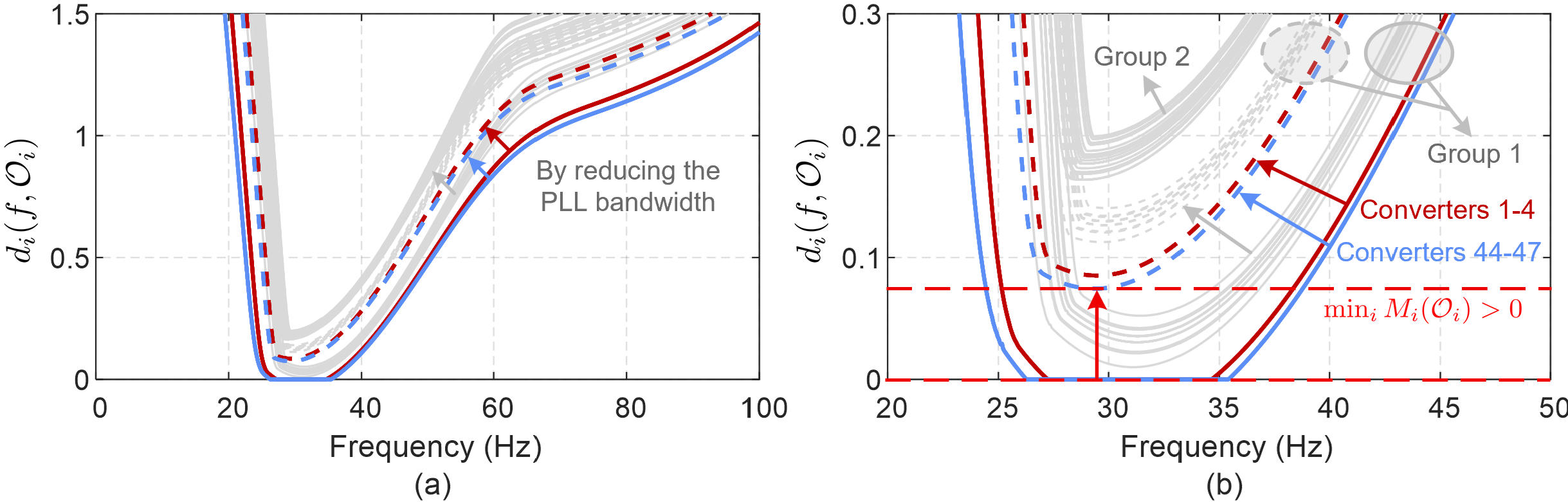}}
	\vspace{-3mm}
    \caption{Frequency-wise distance \(d_i(f,\mathcal O_i)\) of the wind power system after two active-power steps. (a) All 54 converters. (b) Zoom-in view showing that \(M_i(\mathcal O_i)>0\) for all the converters after the PLL modification.}
	\vspace{-4mm}
	\label{frequency_wise_distance}
    \end{figure}

 \section{Conclusions} \label{6}
This paper proposes a geometric decentralized stability certificate for PE-dominated power systems covering variable operating points. The certificate enables modular and scalable stability analysis of heterogeneous multi-converter systems without enumerating the high-dimensional system operating points. A quantitative stability margin based on $x$-$z$ graphs is proposed to construct the decentralized stability operating regions and guide the search for worst-case operating conditions. An efficient algorithm is further proposed for fast evaluation of stability margins and construction of the certified regions. Case studies on a single-converter system and a 54-converter wind power system validate the proposed method.

%TODO REFERANCE
\bibliographystyle{IEEEtran}

%\appendices
%\numberwithin{equation}{section}
%\vspace{0mm}
%\section{Parameters of the Test Systems}
%\label{Appendix C}

\bibliography{REFERANCE}

\newpage

\appendices
\numberwithin{equation}{section}
\vspace{0mm}

\section{Proof of Proposition IV.3}
\label{Appendix A}

\textit{Proof:}
For any fixed \((\omega,V_i)\), let \(\mathbf u_i=(P_i,Q_i)\) and write
\(\mathbf H_i(\omega,\mathbf u_i,V_i):=\mathbf H_i(\omega,\mathcal O_i)\).
From (28) and (30), we obtain
\begin{equation}
\tag{A.1}
\underline{x}_i(\omega,\mathbf u_i,V_i)
=
\min_{\|v\|_2=1}
v^*
\bigl(
\mathbf H_{i,0}
+
P_i\mathbf H_{i,P}
+
Q_i\mathbf H_{i,Q}
\bigr)
v ,
\end{equation}
where \(\mathbf H_{i,0}\), \(\mathbf H_{i,P}\), and \(\mathbf H_{i,Q}\)
are independent of \(\mathbf u_i\), with the fixed \((\omega,V_i)\)
omitted. Thus, for each fixed unit vector \(v\), the expression inside
the minimization is affine in \(\mathbf u_i\).

We first prove the concavity. Take arbitrary power-injection points
\(\mathbf u_i^{(1)}\) and \(\mathbf u_i^{(2)}\), and any
\(\lambda\in[0,1]\). Let
\begin{equation}
\tag{A.2}
\mathbf u_i^{(\lambda)}
=
\lambda \mathbf u_i^{(1)}
+
(1-\lambda)\mathbf u_i^{(2)} .
\end{equation}

For compactness, denote
\(\mathbf H_i^{(k)}=\mathbf H_i(\omega,\mathbf u_i^{(k)},V_i)\),
\(k=1,2\). By the affine structure of \(\mathbf H_i\) in (A.1), we have
\begin{equation}
\tag{A.4}
\begin{aligned}
&\underline{x}_i(\omega,\mathbf u_i^{(\lambda)},V_i) \\
&=
\min_{\|v\|_2=1}
v^*
\left[
\lambda \mathbf H_i^{(1)}
+
(1-\lambda)\mathbf H_i^{(2)}
\right]v \\
&=
\min_{\|v\|_2=1}
\Bigl\{
\lambda v^*\mathbf H_i^{(1)}v
+
(1-\lambda)v^*\mathbf H_i^{(2)}v
\Bigr\} \\
&\ge
\lambda
\min_{\|v\|_2=1}v^*\mathbf H_i^{(1)}v
+
(1-\lambda)
\min_{\|v\|_2=1}v^*\mathbf H_i^{(2)}v \\
&=
\lambda \underline{x}_i(\omega,\mathbf u_i^{(1)},V_i)
+
(1-\lambda)\underline{x}_i(\omega,\mathbf u_i^{(2)},V_i).
\end{aligned}
\end{equation}
The inequality follows since each term is bounded below by its own
minimum and \(\lambda,1-\lambda\ge0\). Hence,
\(\underline{x}_i(\omega,\mathcal O_i)\) is concave with respect to
\(\mathbf u_i=(P_i,Q_i)\) for any fixed \((\omega,V_i)\).

We next prove that the minimum is attained on the boundary. Write \(f_i(\mathbf u_i):=\underline{x}_i(\omega,\mathbf u_i,V_i)\). Since
\(f_i\) is continuous and \(\mathcal U_i\) is nonempty and compact, there
exists \(\mathbf u_i^\star\in\mathcal U_i\) such that
\begin{equation}
\tag{A.6}
f_i(\mathbf u_i^\star)
=
\min_{\mathbf u_i\in\mathcal U_i} f_i(\mathbf u_i).
\end{equation}

If \(\mathbf u_i^\star\in\partial\mathcal U_i\), then (31) follows
directly. Otherwise, \(\mathbf u_i^\star\in\operatorname{int}(\mathcal U_i)\).
Consider a line passing through \(\mathbf u_i^\star\), and let
\([\mathbf u_i^-,\mathbf u_i^+]\) be the connected component of its
intersection with \(\mathcal U_i\) that contains \(\mathbf u_i^\star\).
Then \(\mathbf u_i^-,\mathbf u_i^+\in\partial\mathcal U_i\), and
\begin{equation}
\tag{A.7}
\mathbf u_i^\star
=
\alpha\mathbf u_i^-
+
(1-\alpha)\mathbf u_i^+
\end{equation}
for some \(\alpha\in(0,1)\). By the concavity of \(f_i\),
\begin{equation}
\tag{A.8}
f_i(\mathbf u_i^\star)
\ge
\alpha f_i(\mathbf u_i^-)
+
(1-\alpha)f_i(\mathbf u_i^+).
\end{equation}
Hence,
\(f_i(\mathbf u_i^\star)\ge f_i(\mathbf u_i^-)\) or
\(f_i(\mathbf u_i^\star)\ge f_i(\mathbf u_i^+)\) holds. Since
\(\mathbf u_i^\star\) is a minimizer, equality holds for at least one of
\(\mathbf u_i^-\) and \(\mathbf u_i^+\). Therefore, a minimizer lies on
\(\partial\mathcal U_i\). This proves (31), i.e.,
\begin{equation}
\tag{A.9}
\min_{\mathbf u_i\in\mathcal U_i} f_i(\mathbf u_i)
=
\min_{\mathbf u_i\in\partial\mathcal U_i} f_i(\mathbf u_i).
\end{equation}

Finally, for \(0\le\rho_1<\rho_2\), the nesting condition gives
\(\mathcal U_i(\rho_1)\subseteq\mathcal U_i(\rho_2)\). Hence, the minimum
over \(\mathcal U_i(\rho_2)\) is no larger than that over
\(\mathcal U_i(\rho_1)\). Combining this with (A.9) gives
\begin{equation}
\tag{A.10}
\min_{\mathbf u_i\in\partial\mathcal U_i(\rho_2)}
f_i(\mathbf u_i)
\le
\min_{\mathbf u_i\in\partial\mathcal U_i(\rho_1)}
f_i(\mathbf u_i),
\end{equation}
which is equivalent to (32). This completes the proof.

\section{Parameters of the Test Systems}
\label{Appendix B}
See Table~\ref{table:single converter} and Table~\ref{practical sys_parameters}.
\renewcommand{\arraystretch}{1.25}
\renewcommand{\thetable}{B.\Roman{table}}  
\setcounter{table}{0}

\begin{table}[h]
	\scriptsize
	\centering
	\vspace{-0mm}
	\caption{Parameters of the Single-Converter System}
    \vspace{-1mm}
	\begin{tabular}{|l|l|l|l|}
		\hline
		\multicolumn{4}{|c|}{\bf Base Values for Per-unit Calculation} \\
		\hline
		\multicolumn{1}{|l}{$f_{\rm base} = 50{\rm Hz}$}			
		&\multicolumn{1}{l}{$\omega_{\rm base} = 2\pi f_{\rm base}$}				
        &\multicolumn{1}{l}{$U_{\rm base} = 380{\rm V}$}			
		&\multicolumn{1}{l|}{$S_{\rm base} = 1{\rm kVA}$}								\\
		\hline
		\multicolumn{4}{|c|}{\bf Main Parameters of the PQ Test System (Per-unit Values)}				\\
		\hline
		\multicolumn{2}{|l}{Converter-side inductors: $L_{\rm F} = 0.05$}			
		&\multicolumn{2}{l|}{\textit{LCL} capacitors: $C_{\rm F} = 0.06$}								\\
		\multicolumn{2}{|l}{Grid-side inductors: $L_{\rm g} = 0.417$}			
		&\multicolumn{2}{l|}{Grid-side resistors: $R_{\rm g} = 0.0417$}									\\
        \multicolumn{2}{|l}{PI gains of the current loop: $0.3, 10$}			
		&\multicolumn{2}{l|}{PI gains of the outer loop: $0.5, 40$}	\\
            \multicolumn{2}{|l}{PI gains of the PLL: $88, 3892$}			
		&\multicolumn{2}{l|}{VF time constant: $T_{\rm VF} = 0.005$}	\\	
		\hline

		\multicolumn{4}{|c|}{\bf Main Parameters of the PV Test System (Per-unit Values)}							\\
		\hline
		\multicolumn{2}{|l}{Converter-side inductors: $L_{\rm F} = 0.05$}			
		&\multicolumn{2}{l|}{\textit{LCL} capacitors: $C_{\rm F} = 0.06$}								\\
		\multicolumn{2}{|l}{Grid-side inductors: $L_{\rm g} = 0.556$}			
		&\multicolumn{2}{l|}{Grid-side resistors: $R_{\rm g} = 0.0556$}									\\
        \multicolumn{2}{|l}{PI gains of the current loop: $0.3, 10$}			
		&\multicolumn{2}{l|}{PI gains of the outer loop: $0.5, 40$}	\\
            \multicolumn{2}{|l}{PI gains of the PLL: $43, 931$}			
		&\multicolumn{2}{l|}{VF time constant: $T_{\rm VF} = 0.005$}	\\	
		\hline
        \multicolumn{4}{|c|}{\textbf{Algorithm~1 Parameters for the Test Systems}} \\
        \hline
		\multicolumn{4}{|l|}{$f_{\min}=0$ Hz~~~~$f_{\max}=100$ Hz~~~~$\Delta f=2$ Hz~~~~$N_{\text{l}}=20$~~~~$N_{\text{g}}=40$} \\
		\multicolumn{4}{|l|}{$x_L = -2/L_{\rm g}$~~~$x_R = -1/L_{\rm g}$~~~~~~$\varepsilon=5\times10^{-4}$~~~~~~$\varepsilon_{\text{op}}=5\times10^{-2}$}
		\\
		\hline
	\end{tabular}
    \vspace{-3mm}
	\label{table:single converter}
\end{table}

		%\multicolumn{4}{|l|}{$N_t=20$~~~~~~$N_x=20$~~~~~~$\varepsilon=5\times10^{-4}$~~~~~~$\varepsilon_{\text{op}}=5\times10^{-2}$}
		%\\
        %\multicolumn{4}{|l|}{$x_L = -2/L_{\rm g}$~~~$x_R = -1/L_{\rm g}$~~~$\eta=10^{-8}$~~~$N_{\max}=30$~~~~$\lambda=10^{-6}$}

\begin{table}[H]
	\scriptsize
	\centering
	\vspace{0mm}
    \setlength{\tabcolsep}{4pt}
	\caption{Parameters of the Real-World Wind Power System}
	\begin{tabular}{|llll|}
		\hline
        \multicolumn{4}{|c|}{\bf Base Values for Per-unit Calculation}				\\
        \hline
        $f_{\rm base} = 50{\rm Hz}$		&	$\omega_{\rm base} = 2\pi f_{\rm base}$ &	$U_{\rm base} = 220{\rm kV}$	&	$S_{\rm base} = 1{\rm MVA}$		\\
        \hline
		\multicolumn{4}{|c|}{\bf Main Parameters of the Converters (p.u.)}	\\
        \hline
		\multicolumn{2}{|l}{Converter-side inductors: $L_{\rm F} = 0.05$}			
		&\multicolumn{2}{l|}{\textit{LCL} capacitors: $C_{\rm F} = 0.05$}								\\
        \multicolumn{2}{|l}{PI gains of the current loop: $0.3, 10$}			
		&\multicolumn{2}{l|}{PI gains of the outer loop: $0.5, 40$}	\\
            \multicolumn{2}{|l}{VF time constant: $T_{\rm VF} = 0.005$}	
        &\multicolumn{2}{l|}{Converter rated capacity: $S_N=1.0$}	\\	
		\hline
        \multicolumn{4}{|c|}{\bf Active Power Setpoints of the Converters~(p.u.)} \\	
		\hline
    $P_{1} = 0.50$ & $P_{2} = 0.50$ & $P_{3} = 0.50$ & $P_{4} = 0.50$ \\
    $P_{5} = 0.41$ & $P_{6} = 0.41$ & $P_{7} = 0.51$ & $P_{8} = 0.45$ \\
    $P_{9} = 0.45$ & $P_{10} = 0.36$ & $P_{11} = 0.37$ & $P_{12} = 0.37$ \\
    $P_{13} = 0.52$ & $P_{14} = 0.45$ & $P_{15} = 0.45$ & $P_{16} = 0.77$ \\
    $P_{17} = 0.42$ & $P_{18} = 0.63$ & $P_{19} = 0.63$ & $P_{20} = 0.36$ \\
    $P_{21} = 0.84$ & $P_{22} = 0.71$ & $P_{23} = 0.71$ & $P_{24} = 0.71$ \\
    $P_{25} = 0.71$ & $P_{26} = 0.42$ & $P_{27} = 0.42$ & $P_{28} = 0.62$ \\
    $P_{29} = 0.62$ & $P_{30} = 0.91$ & $P_{31} = 0.91$ & $P_{32} = 0.59$ \\
    $P_{33} = 0.82$ & $P_{34} = 0.82$ & $P_{35} = 0.33$ & $P_{36} = 0.42$ \\
    $P_{37} = 0.67$ & $P_{38} = 0.54$ & $P_{39} = 0.41$ & $P_{40} = 0.41$ \\
    $P_{41} = 0.41$ & $P_{42} = 0.41$ & $P_{43} = 0.28$ & $P_{44} = 0.60$ \\
    $P_{45} = 0.60$ & $P_{46} = 0.60$ & $P_{47} = 0.60$ & $P_{48} = 0.60$ \\
    $P_{49} = 0.60$ & $P_{50} = 0.60$ & $P_{51} = 0.32$ & $P_{52} = 0.45$ \\
    $P_{53} = 0.45$ & $P_{54} = 0.83$ & & \\
        \hline
        \multicolumn{4}{|c|}{\bf Network Parameters of the Practical Wind Power System~($\times 10^{-4}$ p.u.)}				\\	
		\hline
    $X_{55,86} = 1.58$ & $X_{56,86} = 6.55$ & $X_{57,86} = 15.17$ & $X_{58,86} = 21.22$ \\
    $X_{59,86} = 18.32$ & $X_{60,86} = 12.77$ & $X_{61,86} = 10.94$ & $X_{62,86} = 1.58$  \\
    $X_{63,87} = 6.80$ & $X_{64,87} = 20.14$ & $X_{65,87} = 15.01$ & $X_{66,87} = 12.10$ \\
    $X_{67,87} = 9.28$ & $X_{68,87} = 2.24$ & $X_{69,87} = 0.99$ & $X_{70,87} = 0.10$ \\
    $X_{71,88} = 16.33$ & $X_{72,88} = 14.09$ & $X_{73,88} = 9.82$ & $X_{74,88} = 8.42$ \\
    $X_{75,88} = 5.23$ & $X_{76,88} = 8.42$ & $X_{77,88} = 5.23$ & $X_{78,90} = 1.21$  \\
    $X_{79,90} = 14.09$ & $X_{80,89} = 1.21$ & $X_{81,89} = 4.23$ & $X_{82,91} = 9.82$ \\
    $X_{83,91} = 8.42$ & $X_{84,91} = 9.82$ & $X_{85,91} = 8.42$ & $X_{86,87} = 6.70$  \\
    $X_{87,92} = 20.13$ & $X_{88,92} = 27.96$ & $X_{89,90} = 44.64$ & $X_{90,92} = 50.33$ \\
    $X_{91,92} = 5.59$ & $X_{92,93} = 34.77$ & $X_{\text{T}} = 1500.00$ &  \\
        \hline
	\end{tabular}	
    \vspace{0mm}
	\label{practical sys_parameters}
\end{table}

\end{document}